\documentclass[aps,pra,amssymb,superscriptaddress,10pt,tightenlines,twocolumn]{revtex4-2}

\usepackage[english]{babel}
\usepackage{dcolumn}
\usepackage{bm}
\usepackage{amsfonts}
\usepackage{amsmath}
\usepackage{amssymb,amsthm}
\usepackage{mathtools}
\usepackage{braket}
\usepackage{physics}
\usepackage{color}

\usepackage{subfigure}
\usepackage{url}
\usepackage{hyperref}
\usepackage{graphicx}
\usepackage{booktabs}
\usepackage{makecell}
\usepackage{multirow}
\usepackage{comment}

\usepackage[normalem]{ulem}
\usepackage{cancel}

\newtheorem{lem}{Lemma}

\theoremstyle{definition}

\theoremstyle{remark}

\newcommand{\sinc}{\operatorname{sinc}}

\begin{document}

\title{Cooperative photon emission rates in random atomic clouds}
\author{Viviana Viggiano}
\affiliation{Dipartimento di Fisica, Universit\`a di Bari, I-70126 Bari, Italy}
\affiliation{INFN, Sezione di Bari, I-70125 Bari, Italy}
\author{Romain Bachelard}
\affiliation{Universit\'e C\^ote d'Azur, CNRS, Institut de Physique de Nice, 06560 Valbonne, France}
\affiliation{Departamento de F\'isica, Universidade Federal de S\~ao Carlos, Rodovia Washington Lu\'is, km 235-SP-310, 13565-905 S\~ao Carlos, S\~ao Paulo, Brazil}
\author{Fabio Deelan Cunden}
\affiliation{Dipartimento di Matematica, Universit\`a di Bari, I-70125 Bari, Italy}
\affiliation{INFN, Sezione di Bari, I-70125 Bari, Italy}
\author{\\ Paolo Facchi}
\affiliation{Dipartimento di Fisica, Universit\`a di Bari, I-70126 Bari, Italy}
\affiliation{INFN, Sezione di Bari, I-70125 Bari, Italy}
\author{Robin Kaiser}
\affiliation{Universit\'e C\^ote d'Azur, CNRS, Institut de Physique de Nice, 06560 Valbonne, France}
\author{Saverio Pascazio}
\affiliation{Dipartimento di Fisica, Universit\`a di Bari, I-70126 Bari, Italy}
\affiliation{INFN, Sezione di Bari, I-70125 Bari, Italy}
\author{Francesco V. Pepe}
\affiliation{Dipartimento di Fisica, Universit\`a di Bari, I-70126 Bari, Italy}
\affiliation{INFN, Sezione di Bari, I-70125 Bari, Italy}
\date{October 20, 2023}

\begin{abstract}
We investigate the properties of the cooperative decay modes of a cold atomic cloud, characterized by a Gaussian distribution in three dimensions, initially excited by a laser in the linear regime.
We study the properties of the decay rate matrix $S$, whose dimension
coincides with the number of atoms in the cloud, in order to get a deeper insight into properties of cooperative photon emission. Since the atomic positions are random, $S$  is a Euclidean random matrix whose entries are function of the atom distances.
  We show that, in the limit of a large number of atoms in the cloud, the eigenvalue distribution of $S$ depends on a single parameter $b_0$, called the cooperativeness parameter, which can be  viewed as a quantifier of the number of atoms that are coherently involved in an emission process. For very small values of $b_0$, we find that the limit eigenvalue density is approximately triangular. We also study the nearest-neighbour spacing distribution and the eigenvector statistics, finding that, although the decay rate matrices are Euclidean, the bulk of their spectra mostly behaves according to the expectations of classical random matrix theory. In particular, in the bulk there is level repulsion and the eigenvectors are delocalized, therefore exhibiting the universal behaviour of chaotic quantum systems.
\end{abstract}


\maketitle

\section{Introduction}
The study of cooperative effects in atom-photon interactions has attracted great attention in the last decades since the seminal article by Dicke~\cite{Dicke}. A system consisting of two or more atoms interacting with the electromagnetic field behaves very differently compared to a single isolated atom. An example of cooperative effect is the modification of the decay rate of a cold atomic cloud, which is enhanced or suppressed with respect to the decay rate $\Gamma$ of an isolated atom, giving rise to the phenomena of \textit{superradiance} and \textit{subradiance}~\cite{exp_sub,exp_sup,b0,essay,intro,sub_kaiser2021, sub2_kaiser2021,direct2}. Also the atomic transition frequency is affected by the presence of other atoms in the cloud, leading to the so called \textit{collective Lamb shift}~\cite{lamb,lamb2}. Overall, the interaction of photons with atomic ensembles gives rise to a variety of cooperative effects~\cite{review_coop} which can be explained only in terms of a collective behavior rather than the individual (independent) components of the system.
	 
Besides the theoretical interest in comprehending the physical processes underlying cooperative effects, the study of light-matter interfaces~\cite{interfaces} is currently the focus of intense experimental research due to its applications ranging from quantum optics and photonics to quantum information and communication. For instance, superradiance can be exploited in the development of an ultra-narrow linewidth superradiant laser~\cite{laser} or for fast writing and reading of quantum information, while subradiance is particularly useful for quantum memory devices~\cite{quantum_internet,appl_super,appl_sub}. 

In this article, we investigate the properties of the cooperative decay modes of a cold atomic cloud, characterized by a Gaussian distribution in three dimensions, initially excited by a laser in the linear regime. Atomic positions are random and the natural formulation of such a problem is in terms of random decay rate matrices, whose dimension coincides with the number of atoms in the cloud. Since the entries of such matrices depend on the distances between pairs of atoms, they fall in the category of Euclidean random matrices. The main objective of our work is to determine the asymptotic properties of the decay rate matrices spectra. In tackling this problem, it is worth noticing that it is not trivial even to determine the relevant physical parameter(s) on which the asymptotic regime depends. We remark that the Gaussian distribution considered in this article represents a realistic approximation of the experimental situations. A similar system has been experimentally studied for example in~\cite{exp_sub,exp_sup,b0,sub_kaiser2021,sub2_kaiser2021,K}, while a theoretical investigation of the associated random matrices can be found in~\cite{skip_goetschy2,skip_goetschy}.

The article is organized as follows. In Sec.~\ref{sec:systema} we introduce the physical system and set up notation, formulating the problem in terms of Euclidean random matrices (ERM).
The spectrum of the ERM of interest is introduced in Sec.~\ref{sec:spectrum}, where we discuss  various asymptotics regimes. The main scaling limit is analysed in Sec.\
\ref{sec:smallb0} for small values of the so-called cooperativeness parameter, that quantifies the number of atoms that coherently take part in the emission process.
In Sec.~\ref{sec:NNSD} we first review some relevant microscopic statistics used in random matrix theory. Then, we present the results of our numerical study of the local eigenvalue statistics and the delocalization properties of the eigenvectors of the ERM of interest.  We conclude in Sec.~\ref{sec:conclusions} with a summary of the results and some discussion.
Some technical results are reported in Appendices~\ref{app:positive}-~\ref{app:joint}-~\ref{app:mom}.

\section{Physical setting}
\label{sec:systema}

We consider a system of $N$ atoms with the same internal structure, at positions $\vb*{r}_j$, $j=1,\ldots,N$. 
Each atom is approximated as a two-level system with transition frequency $\omega_a$.
The atomic cloud is illuminated by a laser which is described classically by a monochromatic plane wave with wavevector $\vb*{k_0}$, and electric field $\vb*{E_0}$. 
The laser frequency $\omega_0=c\abs{\vb*{k_0}}$ differs from the atomic transition frequency by the detuning $\Delta_0=\omega_0-\omega_a$.
The system under consideration is depicted in Fig.~\ref{fig:system}.

\begin{figure}[h!]
\centering
\includegraphics[scale=0.24]{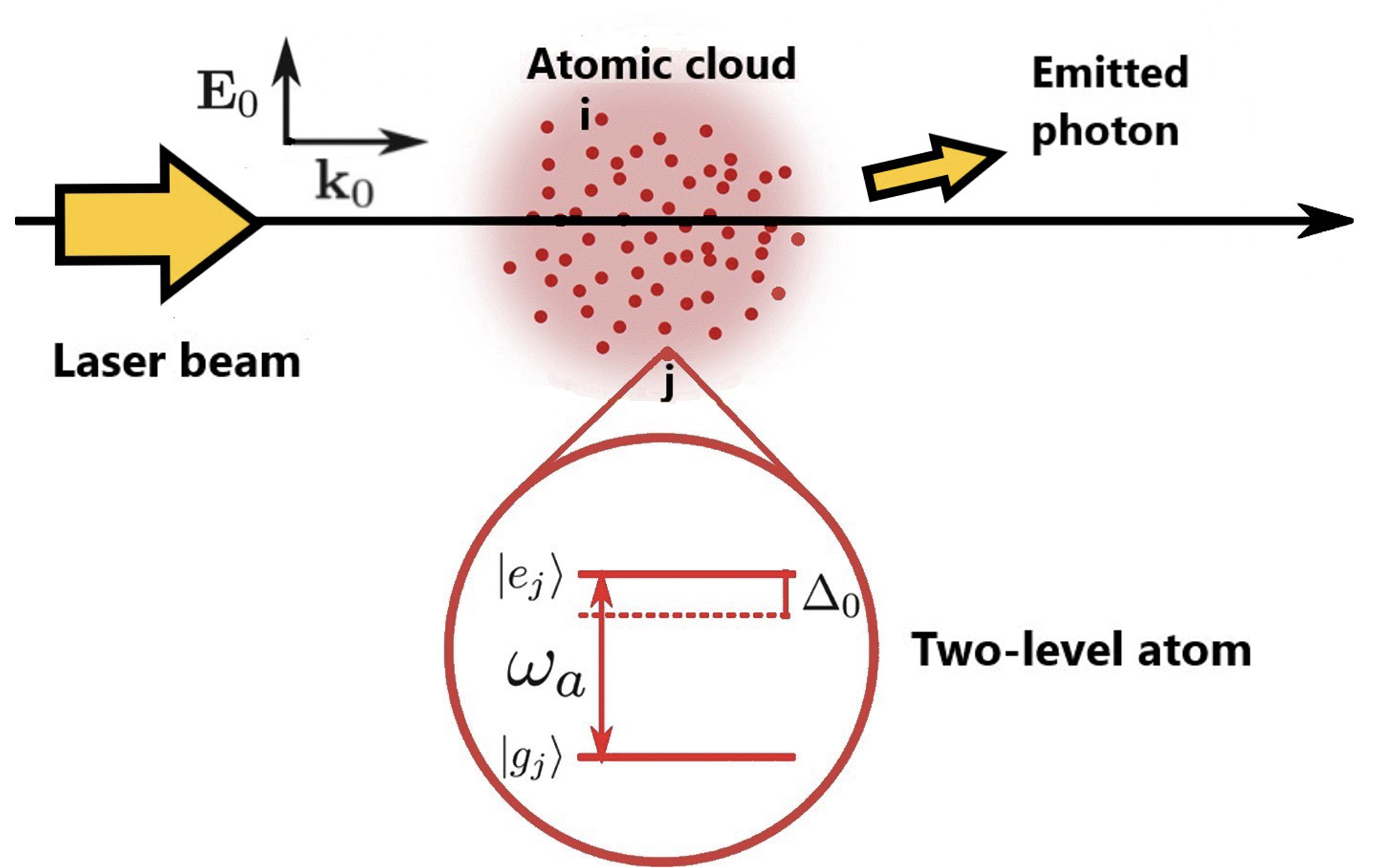}
\caption{Sketch of the system investigated. A Gaussian cloud of two-level atoms with transition frequency $\omega_a$ is illuminated by a laser with wavevector $\vb*{k_0}$, electric field $\vb*{E_0}$ and detuning $\Delta_0$. In the weak excitation limit, the cloud can absorb and emit at most one photon at a time.}
\label{fig:system}
\end{figure}

We assume that the $N$ atoms in the cloud interact with the electric field according to the dipole model in the rotating wave approximation. We will work in the scalar approximation~\cite{Cohen1,Cohen} (for the vectorial case see~\cite{vectorial}) and in the weak excitation limit, the so called \textit{linear regime}, such that the atom-field states can have at most one excited atom or one emitted photon. Furthermore, we assume that the atoms are fixed throughout the whole system evolution, which is a reasonable approximation in the case of large atomic mass and sufficiently low temperature (see Ref.~\cite{Weiss2019} for a study of the robustness of lifetimes against thermal decoherence).

Under these assumptions, denoting by $\beta_j(t)$ the \emph{excitation probability amplitude} of the $j$-th atom ($j=1,\dots,N$), one obtains the following set of coupled equations~\cite{exp_sup, K, Courteille}:
\begin{align}
\begin{aligned}
\dv{\beta_j(t)}{t}= &- \frac{\Gamma}{2}\beta_j(t) -\frac{\Gamma}{2} \sum_{m\neq j} \frac{e^{ik_a r_{jm}}}{ik_a r_{jm}}\beta_m(t) \\ &
-\frac{i\Omega_0}{2} e^{i \vb*{k_0}\vdot\vb* {r}_j}e^{-i\Delta_0 t}\,,
\label{eq:final_betaj}
\end{aligned}
\end{align}
where 
$\Omega_0=-\vb*{E_0} \vdot \vb*d_{eg} /\hbar$ is the Rabi frequency, proportional to the atom dipole moment $\vb*d_{eg}$, 
$\Gamma=\abs{\vb*d_{eg}}^2 \omega_a^3/2 \pi \hbar c^3 \epsilon_0$ is the linewidth of the atomic transition, $r_{jm}=\|\vb*{r}_j - \vb*{r}_m\|$ is the distance between the $j$-th and the $m$-th atom, and
\begin{equation}
  k_a = \frac{\omega_a}{c}.
\end{equation}

We will focus on the evolution equation of the total excitation probability 
$P(t)=\vb*{\beta}^{\dagger}(t)\vb*{\beta}(t)$, where 
\begin{equation}
	\vb*{\beta} = (\beta_1,\beta_2,\dots,\beta_N)^T.
\end{equation}
If $\Omega_0=0$, i.e.\ when the laser is switched off, by using Eq.~\eqref{eq:final_betaj} we get
\begin{align}
\begin{aligned}
\dot{P}(t)&= \dot{\vb*{\beta}}^\dagger(t) \vb*{\beta}(t)+\vb*{\beta}^\dagger(t) \dot{\vb*{\beta}}(t)\\
&=- \Gamma \vb*{\beta}^\dagger(t)S\vb*{\beta}(t)\,,
    \label{eq:Pev}
\end{aligned}
\end{align}
where the \textit{decay rate matrix} $S$, encoding information on the dissipative part of the atom cloud dynamics, is the $N\times N$ matrix with entries
\begin{equation}
S_{ij}=
\mathrm{sinc}\left(k_a \|\vb*{r}_i-\vb*{r}_j\|\right), \quad i,j=1,\ldots, N. 
 \label{eq:f_S}
\end {equation}
Hereafter, $\sinc(x)=(\sin x)/x$ for $x\neq0$, and $\sinc(0)=1$. 

The matrix $S$ defined in Eq.~\eqref{eq:f_S} is real symmetric ($S_{ij}=S_{ji}$) and with fixed trace ($S_{ii}=1$),
\begin{equation}\label{eq:trace}
    \operatorname{Tr}S=N.
\end{equation}  
It is also possible to show that $S$ is positive semi-definite. See Appendix~\ref{app:positive}. 
Therefore, the eigenvalues $\lambda_1,\ldots,\lambda_N$ of $S$ satisfy  $0 \leq \lambda_i \leq N$, and $(1/N) \sum_i\lambda_i=1$. In particular, $\lambda=1$ coincides with the decay rate $\Gamma$ of an isolated atom, while $0 \leq \lambda<1$ are associated with subradiant modes characterized by a slower decay rate, and $1<\lambda \leq N$ 
correspond to superradiant modes, decaying with a rate larger than $\Gamma$~\cite{spettroS}.

As a model of the atomic cloud, 
we assume that the random positions $\vb*{r}_i$'s of the $N$ atoms are independent and identically distributed according to a three-dimensional Gaussian probability density function with zero mean  and variance $\sigma^2>0$. 
Under this assumption, the matrix $S$ in~\eqref{eq:f_S} is a Euclidean random matrix (ERM): the matrix entries are given by a deterministic function $f(\|\vb*{r}_i-\vb*{r}_j\|)= \sinc(k_a \|{\vb*{r}_i-\vb*{r}_j}\|) $ of the Euclidean distances
between the random points $\vb*{r}_i$. See~\cite{parisi,Bogomolny03,Parisi06,skip_goetschy2} and references therein. The matrix $S$ of this paper (with different distribution on the positions $\vb*{r}_i's$) was also studied in~\cite{skip_goetschy}. It is worth noticing that the Gaussian shape provides a good model of a cold atomic cloud in the case of harmonic trapping at finite temperature and low density. Moreover, in many experiments, atoms are released from the trap in order to be probed, and the Gaussian velocity distribution transforms into a Gaussian spatial distribution after some time of flight.

The general problem of ERM theory is to understand the statistical properties of the eigenvalues
and the corresponding eigenvectors in the large-$N$ limit.
This is certainly a very rich and challenging problem. Many interesting results on the eigenvalue distributions are given in Refs.~\cite{spettroS,skip_goetschy2,skip_goetschy}. In the present work, we will start from characterizing the spectrum, and then focus on the microscopic statistics, analyzing level spacing and eigenvector properties. Note that the set of $N\times N$ real symmetric matrices is a linear space of dimension $N(N+1)/2$. The matrix $S$ is instead specified by the $3N$ coordinates of the $N$ atoms, and $3N\ll N(N + 1)/2$ for large $N$.

\section{Spectrum of the decay rate matrix $S$}

\label{sec:spectrum}

\subsection{General aspects}\label{sec:specs}

The main goal of our analysis is to determine the properties of the random matrix $S$ in the large-$N$ limit, i.e.\ in the asymptotic regime of large number of atoms ($N \approx 10^9$ atoms in experiments \cite{exp_sub,exp_sup}).
It is convenient to write the positions of the atoms as $\vb*{r}_i=\sigma\vb*{x}_i$, where  $\vb*{x}_1,\ldots,\vb*{x}_N$ are independent random variables distributed according to the three-dimensional standard Gaussian density 
\begin{equation}
p_X(\vb*{x})=\frac{1}{(2\pi)^{3/2}} e^{-x^2/2},\quad \vb*{x} \in \mathbb{R}^3, \quad x=\|{\vb*{x}}\|.
    \label{eq:gaussian}
\end{equation}
The average density of the atomic cloud is 
\begin{equation}
 \rho=\frac{N}{\sigma^3},
    \label{eq:rho}
\end{equation}
and, in terms of the dimensionless variables $\vb*{x}_i$'s, the matrix $S$ is
\begin{equation}
S_{ij}=  \sinc\left(\sqrt{M}\|{\vb*{x}_i - \vb*{x}_j}\| \right), \quad i,j=1,\ldots, N,
 \label{eq:S2}
\end {equation}
where 
\begin{equation}
  M=(k_a\sigma)^2
  \label{eq:Mdef}
\end{equation}
is a dimensionless parameter that plays a key role in the large-$N$ limit, since it represents an estimate of the number of independent transverse optical modes in the cloud.

We stress that the entries of the random matrix $S$ are \emph{not} independent, but Euclidean-correlated. The diagonal entries are nonrandom $S_{ii}=1$,  while the off-diagonal entries are random variables $0\leq \abs{S_{ij}}\leq 1$.

 If $\vb*{x}_i,\vb*{x}_j$ are two independent standard Gaussian points, the probability density of their distance $R=\|\vb*{x}_i-\vb*{x}_j\|$, $i\neq j$ is 
\begin{equation}
p_R(r)=\frac{1}{\sqrt{4\pi}}r^2e^{-r^2/4},\quad r\geq0.
\label{eq:density_R}
\end{equation}

If $\vb*{x}_i,\vb*{x}_j,\vb*{x}_k,\vb*{x}_l$ are distinct, then the distances $R_{ij}=\|\vb*{x}_i-\vb*{x}_j\|$ and $R_{kl}=\|\vb*{x}_k-\vb*{x}_l\|$ are independent. Therefore, the matrix entries $S_{ij}$ and $S_{kl}$ are independent unless they share one index, while entries of $S$ that share at least one index (e.g. entries on the same column/row) are Euclidean-correlated. For instance $S_{ij}$ and $S_{il}$, with $i,j,l$ all distinct, depend on the interatomic distances $R=\|\vb*{x}_i-\vb*{x}_j\|$ and $R'=\|\vb*{x}_i-\vb*{x}_l\|$ whose joint density is \emph{not} factorized and reads
\begin{equation}
  p_{R,R'}(r,r')=\frac{2}{\sqrt{3} \pi } r e^{-\frac{r^2}{3}} r' e^{-\frac{r'^2}{3}} \sinh \left(\frac{r r'}{3}\right),\quad r,r'\geq0.
  \label{eq:density_RR'}
\end{equation}
 For a proof of~\eqref{eq:density_R} and~\eqref{eq:density_RR'}, see Appendix~\ref{app:joint}.  
\par

One of the most peculiar features of the matrix $S$ in~\eqref{eq:S2} is that its eigenvalues strongly depend on the parameter $M=(k_a\sigma)^2$. This dependence originates from the nonlinearity of the $\sinc$ function. This can be anticipated by looking at the statistical properties of the off-diagonal \emph{entries} of $S$. 
Using the interatomic distance density function~\eqref{eq:density_R} we can compute the moments of $S_{ij}$, $i\neq j$, and their large-$M$ asymptotics,
\begin{equation}
\langle S_{ij}\rangle=e^{-M}
,\quad
\langle S_{ij}^2\rangle=e^{-2M}\frac{\sinh(2M)}{2M}{\sim} \frac{1}{4M},
\label{eq:mom12}
\end{equation}
where the average $\langle\ldots \rangle$ is taken according to the distribution (\ref{eq:gaussian}).
On the other hand, \emph{all} moments for $m\geq3$ scale as $\langle S_{ij}^m\rangle{\sim}a_mM^{-3/2}$, for large $M$, 
where the $a_m$'s are explicit constants (see Appendix~\ref{app:mom}).
Using the joint density~\eqref{eq:density_RR'} we find instead that, for $i,j,l$ distinct,
\begin{equation}
\langle S_{ij}S_{il}\rangle=e^{-2M}\frac{\sinh(M)}{M}\sim \frac{e^{-M}}{2M}.
\end{equation}

Following previous literature~\cite{exp_sub,exp_sup,b0}, we introduce the \textit{cooperativeness parameter}
\begin{equation}
 b_0=\frac{N}{M} = \frac{N}{(k_a \sigma)^2},
    \label{eq:b0}
\end{equation}
that can be thought of, in the large-$N$ limit and for a large number of modes in the cloud, as the number of atoms that can coherently interact to produce cooperative decay.

\subsection{Limiting cases for fixed $N$}\label{sec:fixedN}

For fixed size $N$, we have that when $M\to\infty$ (small cooperativeness parameter $b_0\ll1$), the matrix elements tend to $S_{ij}\to\delta_{ij}$. Hence, the matrix $S$ tends to the $N\times N$ identity matrix $I$. In this limit, the atoms are too far apart to cooperate in photon emission, and decay independently with the rate $\Gamma$ of an isolated atom.

 On the other hand, the limit $M\to0$ (large cooperativeness parameter $b_0\gg1$) describes $N$ atoms separated by a small distance compared to the emission wavelength. This is the so called \textit{Dicke limit}~\cite{Dicke}, with $N-1$ subradiant modes and only one supperradiant symmetric state. 
 This behaviour immediately follows from the limit $S_{ij}\to 1$, as $M\to0$. Equivalently, the matrix $S$ tend to the $N\times N$  rank-$1$ matrix with all ones, 
whose only non-vanishing eigenvalue $\lambda=N$ corresponds to the superradiant mode, 
$\vb*{\beta} = (1/\sqrt{N})(1,1,\dots,1)^T$, namely the symmetric combination of atomic excitations
\begin{equation}
\ket{\psi^+}= \frac{1}{\sqrt{N}} \sum_{j=1}^N \ket{j} 
\label{eq:symm}\,.
\end{equation}
The $(N-1)$-degenerate subradiant eigenspace corresponding to the zero eigenvalue is composed of states in the form
\begin{equation}\label{eq:psim}
\ket{\psi^-}= \sum_{j=1}^N \beta_j \ket{j} \quad \text{with } \sum_{j=1}^N \beta_j = 0 ,
\end{equation}
orthogonal to the superradiant mode, where the excitation amplitude are combined in such a way to produce destructive interference of the emitted photon, thus hindering decay.

\subsection{Large-$N$ limit with fixed $M$}\label{sec:fixedM}
Suppose now that we wish to analyse the spectrum of the random matrix $S$ as $N\to\infty$ with $M$ fixed. This limit is well understood mathematically. From general results~\cite[Theorem 1]{Jiang15} it follows that all but a negligible fraction of the eigenvalues of $S$ converge to $0$, almost surely. In this limit most of the eigenvalues are close to zero. However, since $\Tr S=N$, it is natural to ask whether by centring and rescaling the eigenvalues of $S$ one can get a nontrivial limit. 
Indeed this is the case and a finer analysis~\cite[Theorem 3.1]{Koltchinskii} shows that the spectrum of $(S-I)/N$ is asymptotically close, as $N\to\infty$, to the spectrum of the \emph{integral operator} $\mathcal{S}$ defined by the formula
\begin{equation}
\mathcal{S}g(\vb*{y})=\int_{\mathbb{R}^3}\sinc\left(\sqrt{M}\|\vb*{x}-\vb*{y}\|\right)g(\vb*{x})p_X(\vb*{x})d\vb*{x}.
    \label{eq:mathcalS}
\end{equation} 
where $p_X$ is defined in Eq.\ (\ref{eq:gaussian}).
The (non-random) operator $\mathcal{S}$ is Hilbert-Schmidt and self-adjoint. Hence, its nonzero eigenvalues have finite multiplicities and form a real sequence convergent to $0$.
Notice that the condition $N\to\infty$, with $M$ fixed, is a high-density limit $\rho\sim N$.

\subsection{Large-$N,M$ limit with fixed cooperativeness parameter}\label{sec:fixedb0}

In order to describe the spectral properties of $S$, we consider the normalized (mean) \emph{eigenvalue density} 
\begin{equation}
p(\lambda)=\left\langle\frac{1}{N}\sum_{i}\delta(\lambda-\lambda_i)\right\rangle,
\end{equation}
 where $\lambda_1,\ldots,\lambda_N$ are the eigenvalues of $S$. The moments of  $p(\lambda)$ are given by
 \begin{equation}
 \langle\lambda^m\rangle=\int_{0}^{+\infty}\lambda^mp(\lambda)d\lambda=\frac{1}{N}\langle\operatorname{Tr}S^m\rangle.
 \end{equation}
 When $m=1$,  from the identity $\frac{1}{N}\operatorname{Tr}S=1$ we get that $p(\lambda)$ has mean $\langle \lambda\rangle=1$ for all values of $N,M$.

 Using~\eqref{eq:mom12} we find for the second moment  
 \begin{align}
\langle\lambda^2\rangle&= \frac{1}{N}\langle\Tr S^2\rangle=\frac{1}{N}\sum_{i,j}  \langle S_{ij}^2\rangle \nonumber \\
&=1+(N-1)e^{-2M}\frac{\sinh(2M)}{2M}\stackrel{N,M\to\infty}{\sim} 1+\frac{N}{4M}.
\end{align}
We see that, if we want to let $N\to\infty$ and keep the second moment finite, we must take $M\to\infty$  with $b_0=N/M$ fixed, so that
  \begin{equation}
  \label{eq:2ndmomlim}
 \langle\lambda^2\rangle\to1+\frac{b_0}{4}.
 \end{equation}
 Remarkably, $b_0$ is also the relevant parameter from a \emph{physical} point of view, see Eq.\ (\ref{eq:b0})
 and Refs.\ \cite{exp_sub,exp_sup,b0}. In principle, all moments of the eigenvalue density $p(\lambda)$  can
be computed with this kind of approach, although the combinatorics is rather difficult. 
 
For a numerical demonstration of the correctness of this asymptotic regime $N,M\to\infty$, with $N/M=b_0$ fixed, we invite the reader to have a glance at Fig.~\ref{fig:sovrapposti} where we compare the eigenvalue distribution obtained by numerically diagonalization of samples of $S$ of size $N=300,1000,10000$, with $b_0=1$. We see that the histograms overlap, thus confirming the convergence of the eigenvalue density in the chosen scaling limit.

In the following we will study the spectral properties of the ERM $S$, by focusing
on this most interesting asymptotic regime obtained by the \emph{simultaneous} limit $N,M\to\infty$ with the ratio
\begin{equation}
 b_0=\frac{N}{M}
    \label{eq:b0_2}
\end{equation}
kept fixed. In other words, we set $\sqrt{M}=\sqrt{N/b_0}$ in~\eqref{eq:S2}, namely,
\begin{equation}
S_{ij}=  \sinc\left(\sqrt{\frac{N}{b_0}}\|{\vb*{x}_i - \vb*{x}_j}\| \right), \quad i,j=1,\ldots, N,
 \label{eq:Sb0}
\end {equation}
and investigate the limit $N\to\infty$ for various values of $b_0$.

It is worth noticing that condition~\eqref{eq:b0_2} amounts to say that $\sigma= M^{1/2}/k_a\sim N^{1/2}$, so that the  density $\rho$ in Eq.~\eqref{eq:rho} scales as $N^{-1/2}$, which is a low-density limit. This should be compared with the standard thermodynamic limit $N\to\infty$ at fixed density, that would instead require a stronger confinement of the atomic cloud, with $\sigma\sim N^{1/3}$.

\begin{figure}[htp]
\centering
\includegraphics[width=.8\columnwidth]{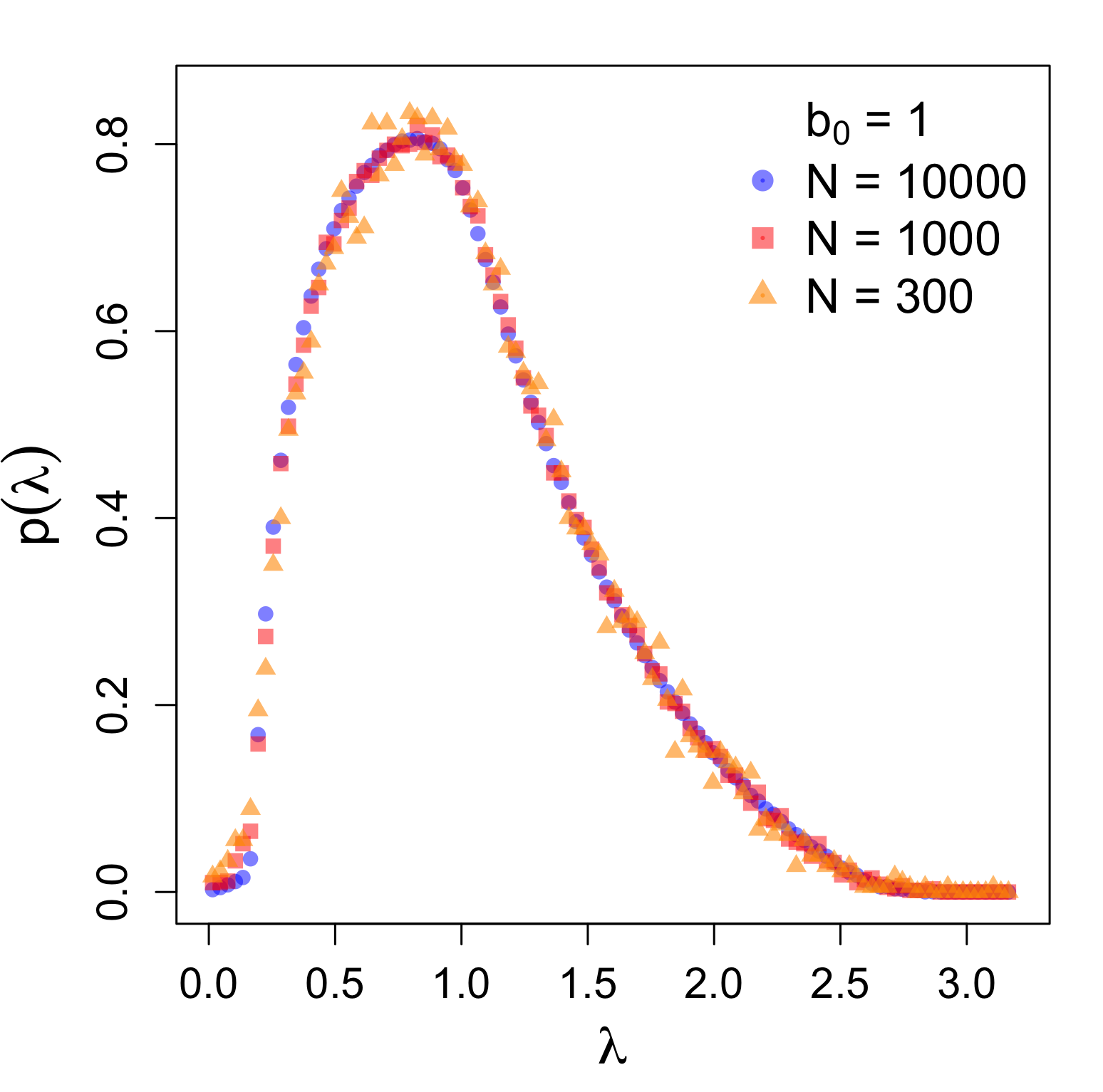}
\caption{Eigenvalue density of $S$  with $b_0=1$ fixed for different matrix sizes $N$.}
\label{fig:sovrapposti}
\end{figure}

\section{The asymptotic eigenvalue density}

\label{sec:smallb0}
\begin{figure*}
\centering
\subfigure[]{\includegraphics[trim={0 0 0 1cm}, clip,width=0.5\columnwidth]{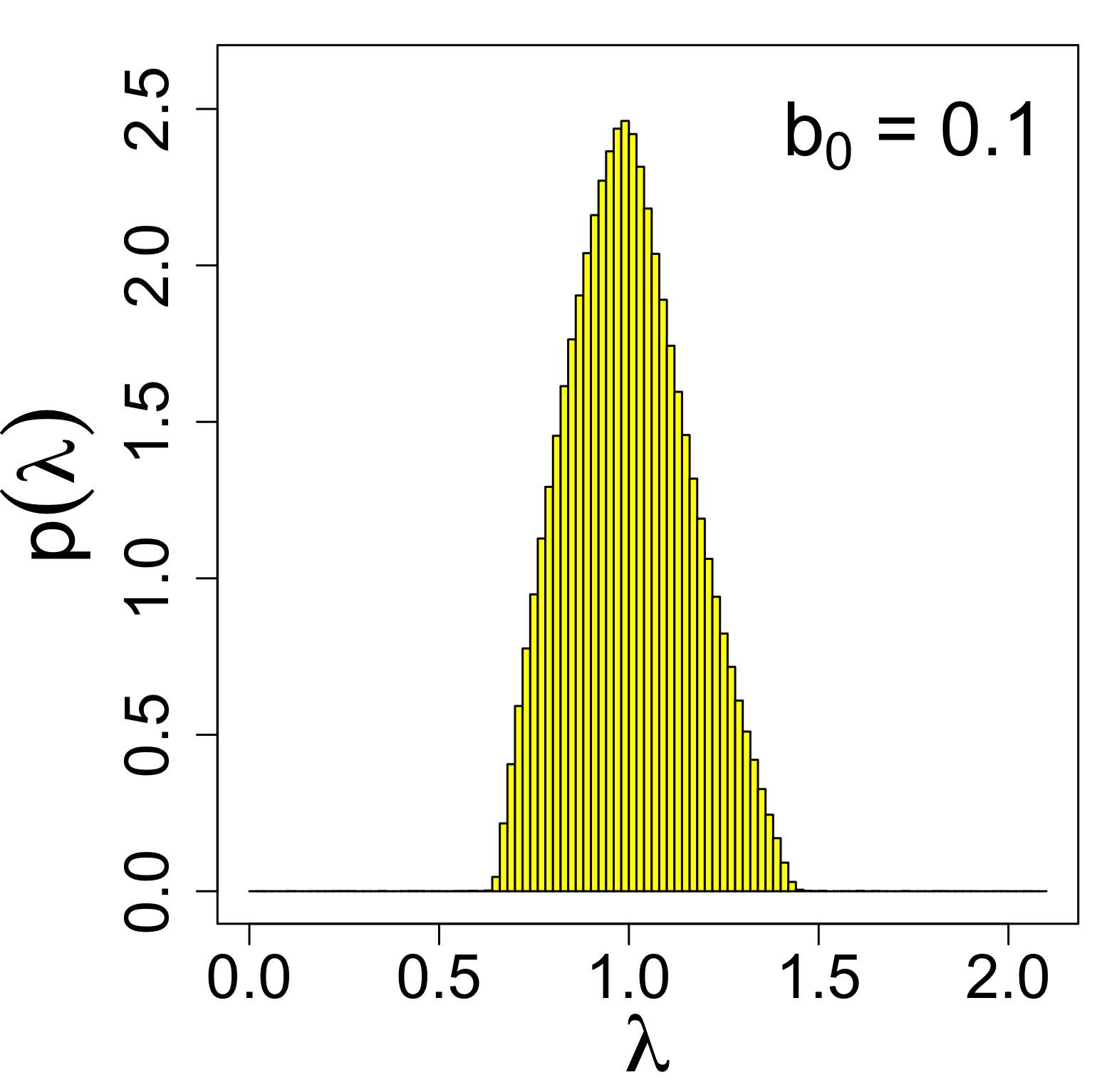}\label{fig:b0_1}}
\subfigure[]{\includegraphics[trim={0 0 0 1cm}, clip,width=0.5\columnwidth]{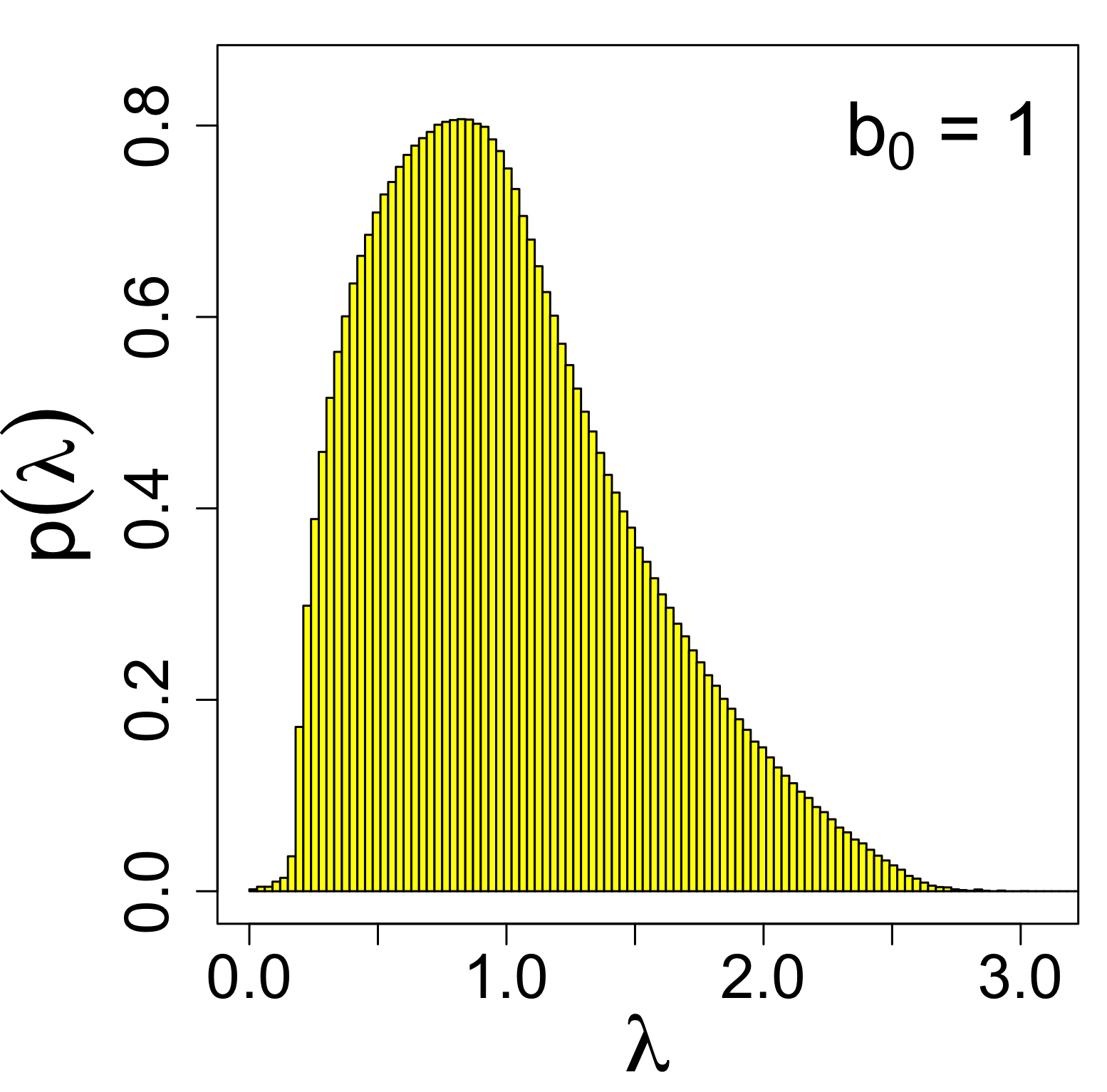}\label{fig:b1}}
\subfigure[]{\includegraphics[trim={0 0 0 1cm}, clip, width=0.5\columnwidth]{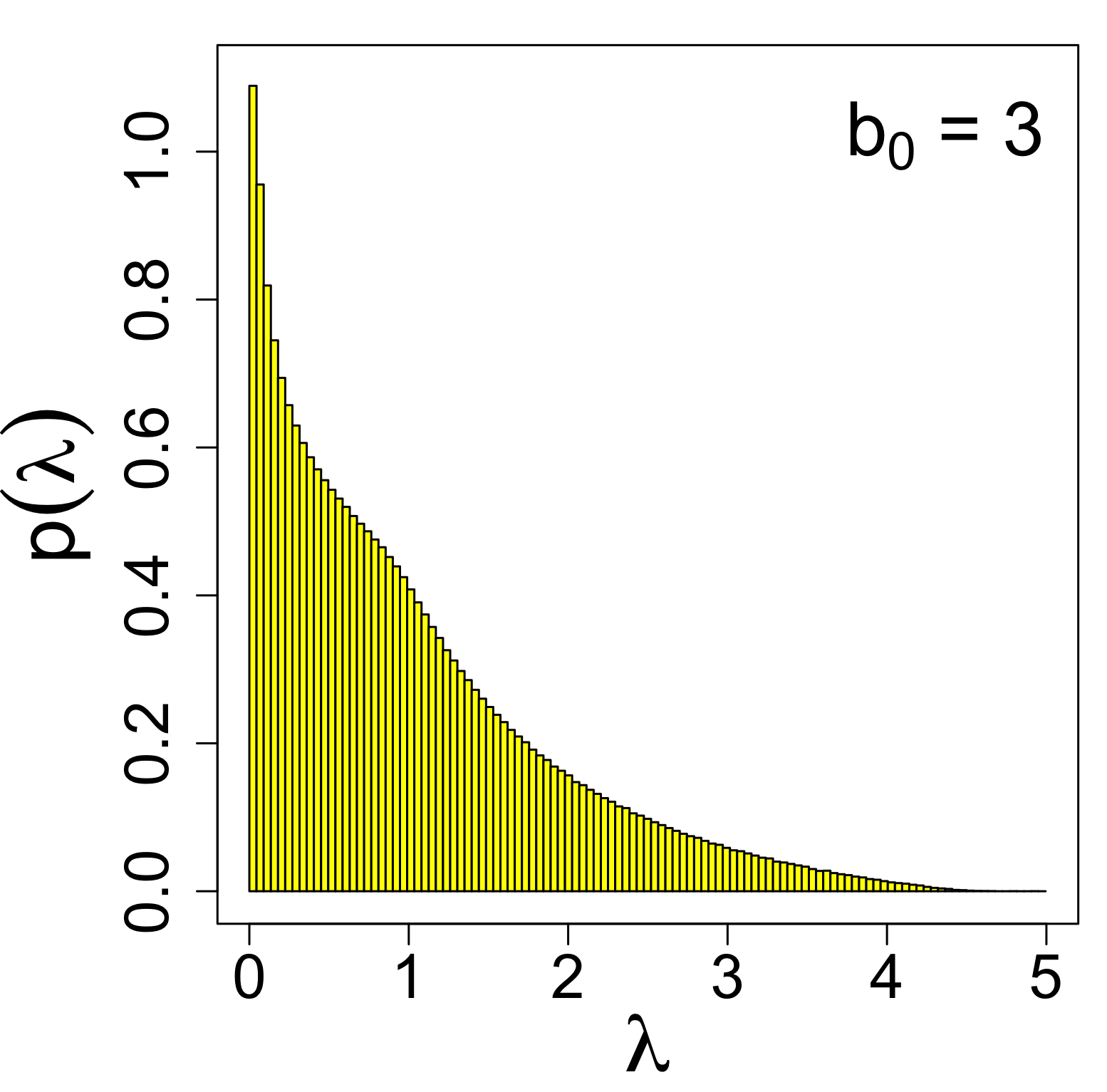}\label{fig:b3}}
\subfigure[]{\includegraphics[trim={0 0 0 1cm}, clip,width=0.5\columnwidth]{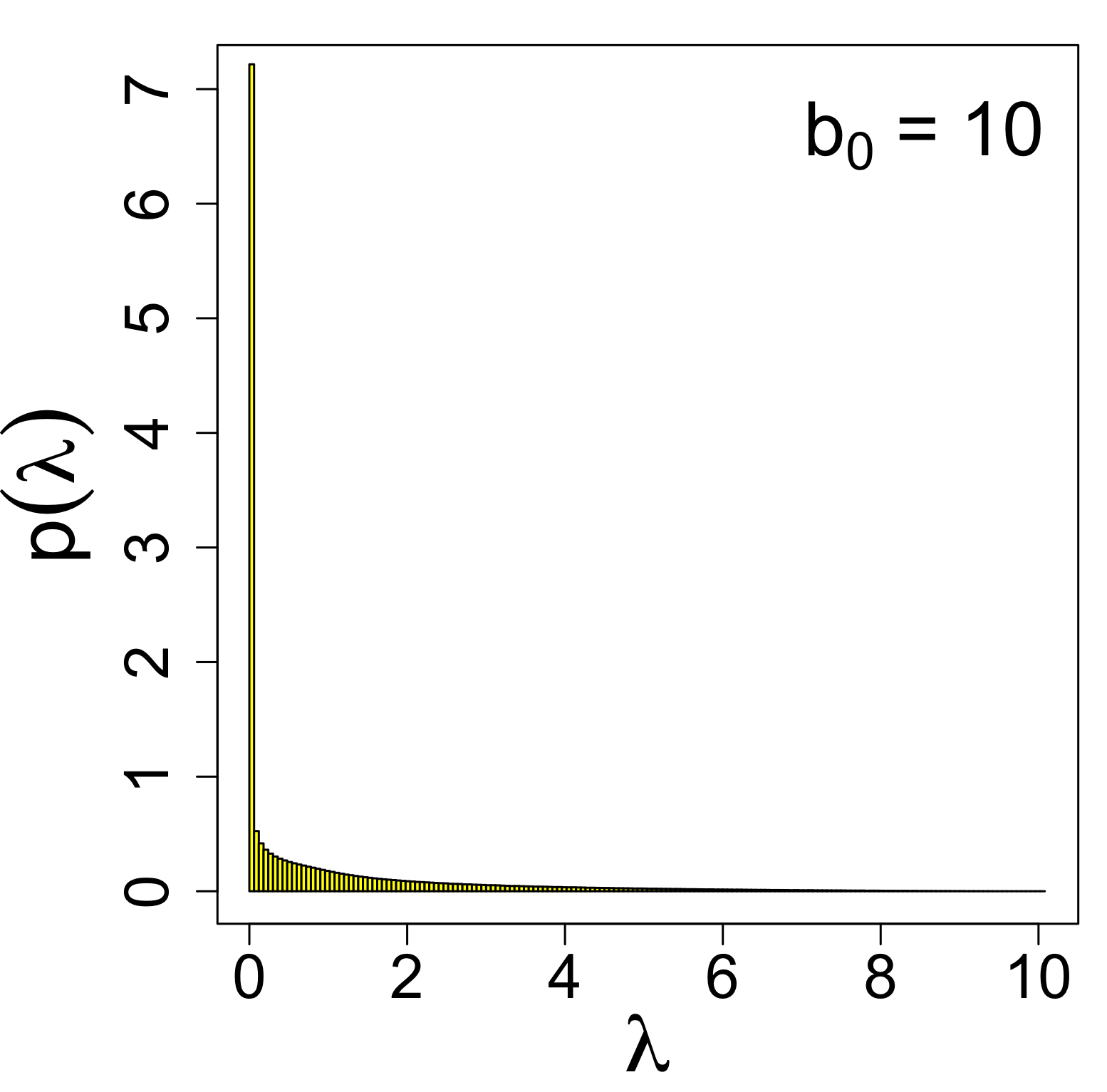}\label{fig:b10}}\\
\subfigure[]{\includegraphics[trim={0 0 0 1cm}, clip,width=0.5\columnwidth]{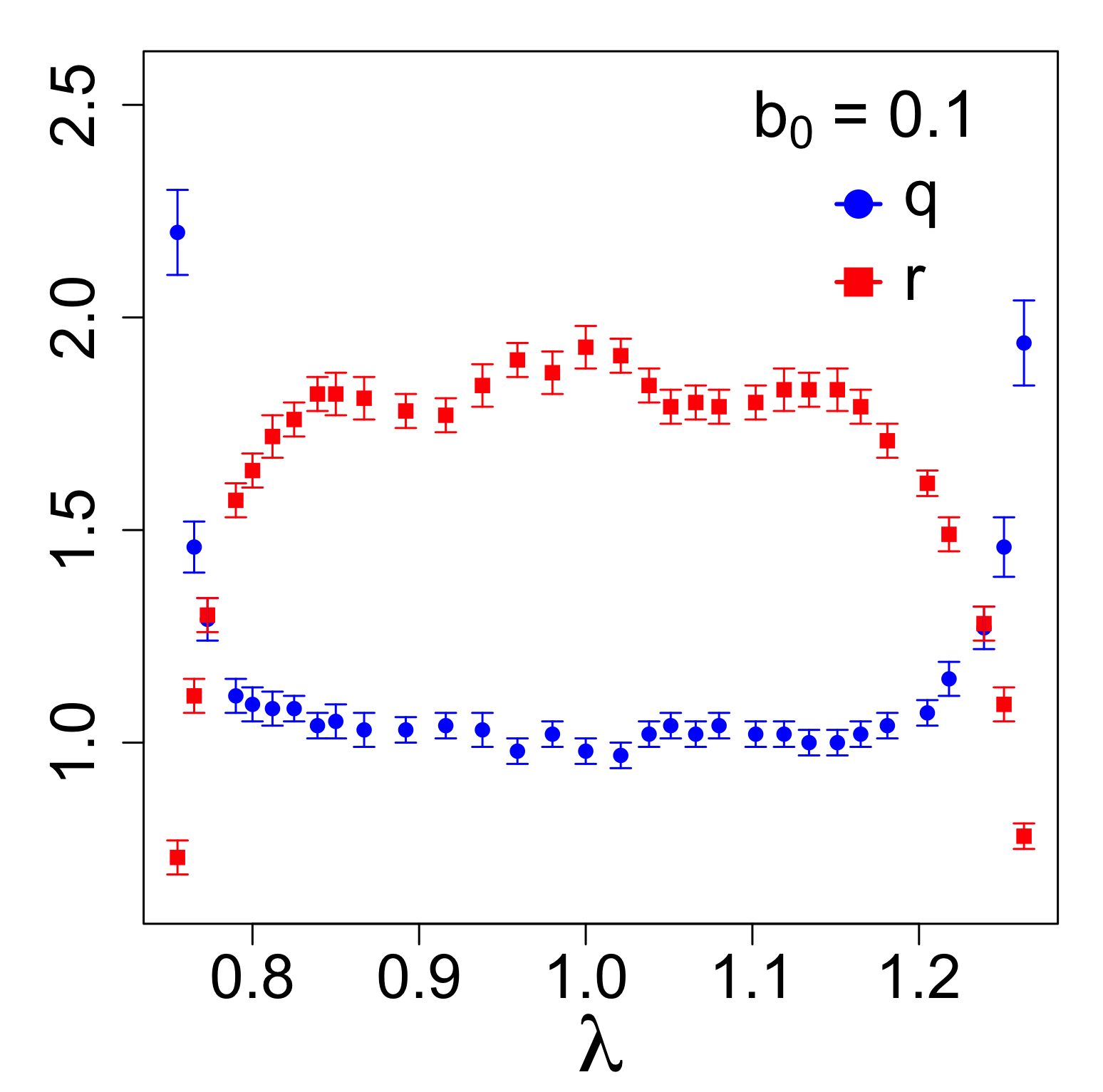}\label{fig:b0_1_sp}}
\subfigure[]{\includegraphics[trim={0 0 0 1cm}, clip,width=0.5\columnwidth]{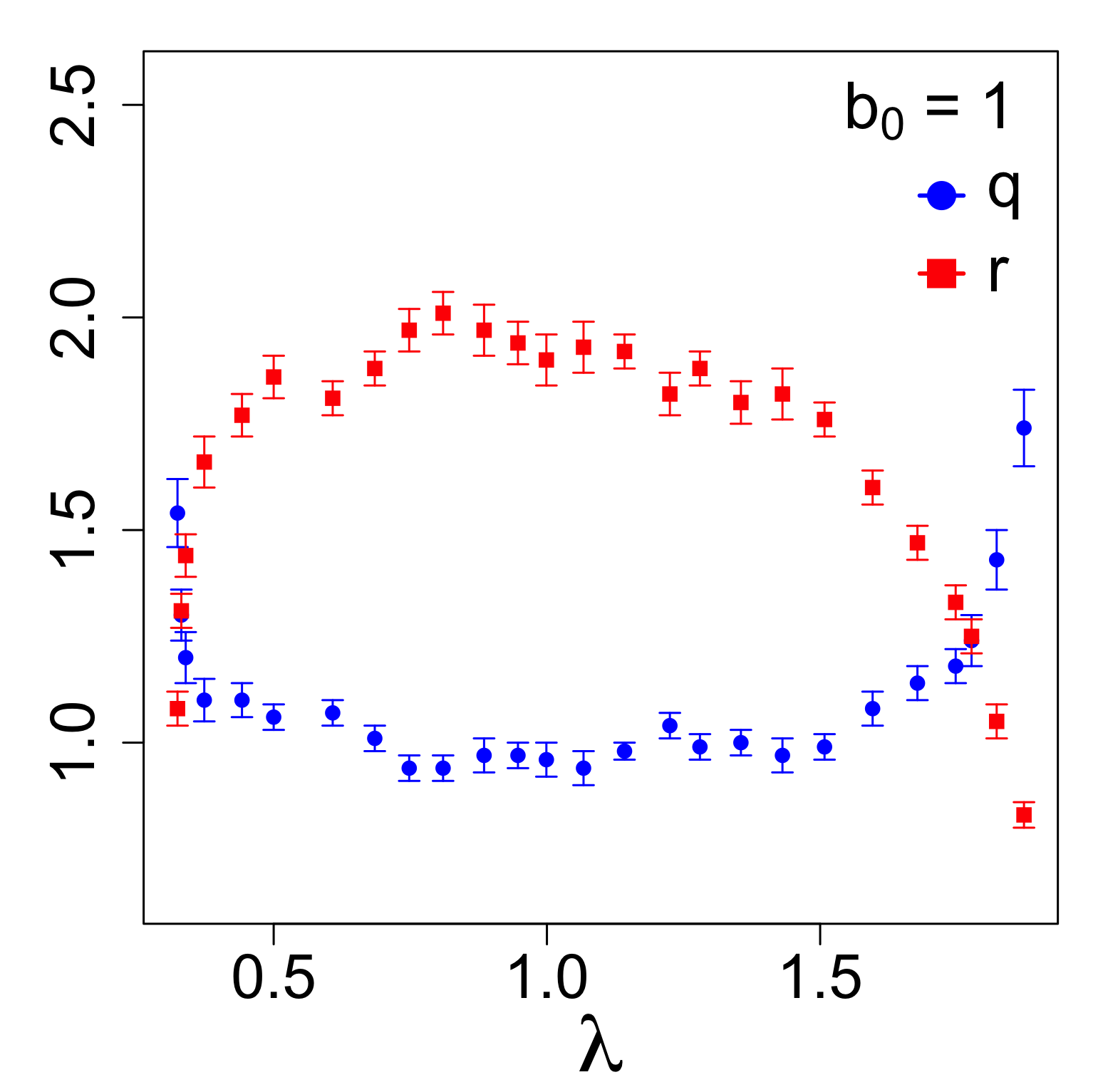}\label{fig:b1_sp}}
\subfigure[]{\includegraphics[trim={0 0 0 1cm}, clip,width=0.5\columnwidth]{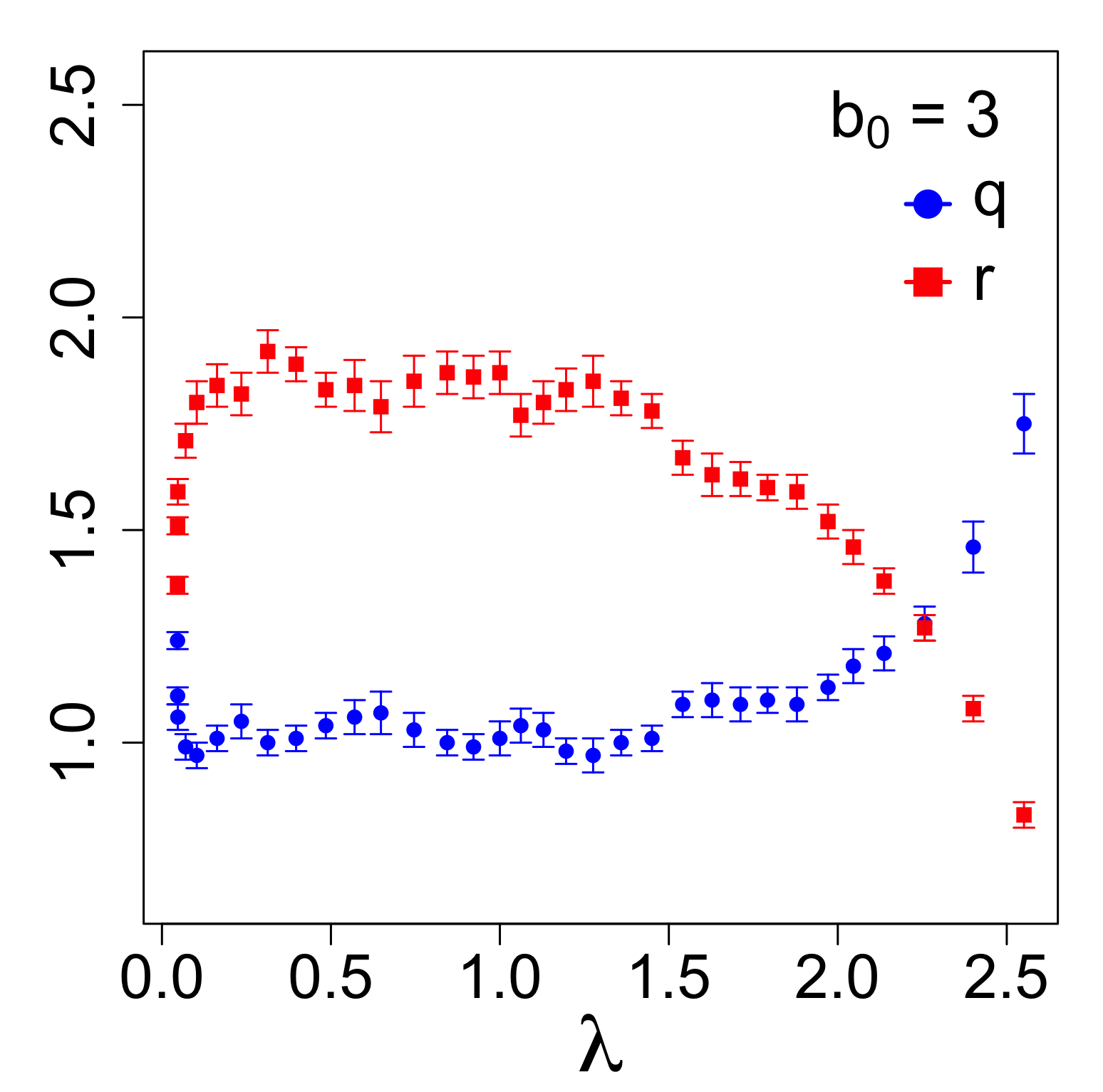}\label{fig:b3_sp}}
\subfigure[]{\includegraphics[trim={0 0 0 1cm}, clip,width=0.5\columnwidth]{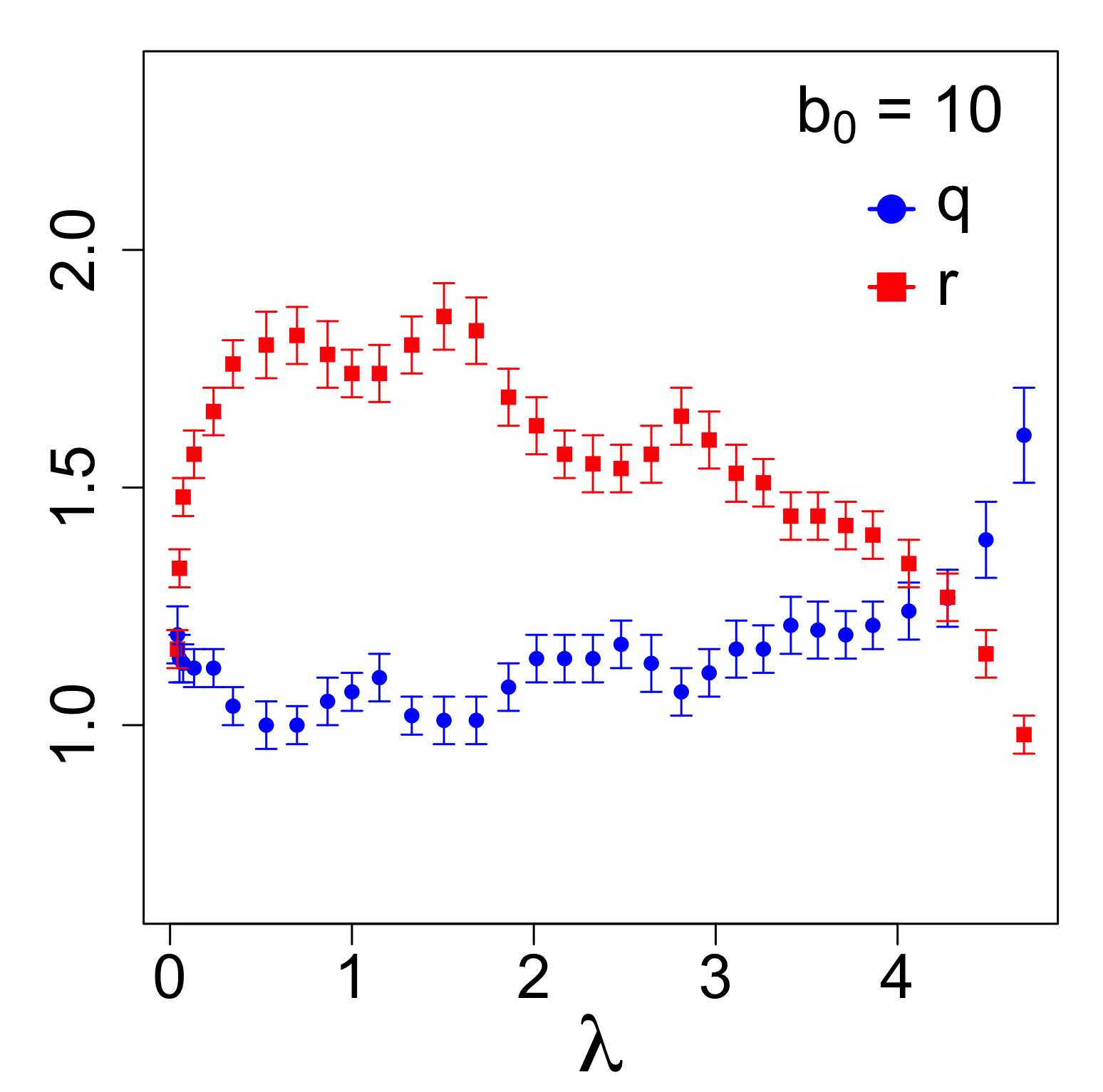}\label{fig:b10_sp}}
\caption{Top: Histograms of the eigenvalue density $p(\lambda)$ obtained from 20 samplings of the matrix  $S$ defined in Eq.~\eqref{eq:Sb0} for different values of the cooperativeness parameter $b_0$. In all simulations $N=10000$.
Bottom: Best-fit parameters $q$ and $r$ of the Wigner-like surmise for the spacings, $ p(q,r;s)=a s^q \, e^{-b s^{r}}$ defined in \eqref{eq:2par}, obtained for different regions of the spectrum containing 1000 eigenvalues (600 eigenvalues for $b_0=10$ in Fig.\ \ref{fig:b10_sp}) centered around $\lambda$. Notice that $q\simeq1$ in the bulk, corresponding to level repulsion.}
\label{fig:eigenvalues}
\end{figure*}

In what follows we study the spectral properties of $S$ as the cooperativeness parameter $b_0$ varies in order to unveil cooperative effects in different regimes. We will find that  large values of $b_0$ enhance the super- and subradiant part of the spectrum of $S$, whereas when $b_0\to0$ the spectrum concentrates at $\lambda=1$. This is in accord with physical intuition.

For some representative values of the cooperativeness parameter $b_0$ we have sampled random matrices $S$ with size $N=10000$. For each value of $b_0$ we performed 20 realizations. The mean eigenvalue density $p(\lambda)$ is estimated by the histogram of the eigenvalues of $S$ obtained by numerical diagonalization,  see Fig.~\ref{fig:eigenvalues}. All histograms are normalized and have unit mean. 

Figure~\ref{fig:b0_1} shows that for small values of the cooperativeness parameter $b_0$ the spectrum is strongly peaked around $\lambda=1$. 
Figs.~\ref{fig:b1}, \ref{fig:b3} and~\ref{fig:b10} show that, as $b_0$ increases, the spectrum of $S$ broadens as expected:  cooperative effects kick in and cause deviations from the isolated-atom decay rate. 
Since the eigenvalues are non-negative and have unit mean, it follows that as $b_0$ increases the shape of the spectrum becomes increasingly asymmetrical due to the accumulation of eigenvalues in the subradiant sector $\lambda <1$. 
Figure~\ref{fig:b10} clearly shows that the fundamental features of the spectrum for large values of $b_0$ are the accumulation of the eigenvalues close to zero and the presence of few large superradiant eigenvalues.

As noted above, most of the techniques developed in random matrix theory for matrices with independent entries or unitarily invariant distributions, are not directly applicable to Euclidean random matrices. Therefore, computing the large-$N$ limit of the mean eigenvalue density of $S$ for generic values of $b_0$ is a challenging open problem. 

We now analyze the eigenvalue density $p(\lambda)$ of $S$ in~\eqref{eq:Sb0} for $b_0\ll1$,
\begin{equation}
\label{eq:plambdab0}
\lim_{b_0\to0}\lim_{N \to\infty}p(\lambda).
\end{equation}
We stress the fact that the two limits do \emph{not} commute. The precise description of~\eqref{eq:plambdab0} is somewhat lengthy, and goes beyond the scope of this paper. We found however a simpler and neat approximation that we now present.
The shape of the histograms for  small values of $b_0$, see Fig.~\ref{fig:trianglefit}, suggests that the spectrum could be described by a triangular density around the single-atom emission peak at $\lambda=1$. 

Indeed, for small $b_0$ one can imagine that $S$ in~\eqref{eq:Sb0} can be approximated by a direct sum of $2\times2$ matrices of the form 
\begin{equation}
\begin{pmatrix}
1 & \sinc(\sqrt{N/b_0} \|\vb*{x}-\vb*{x'}\|)\\
\sinc(\sqrt{N/b_0} \|\vb*{x'}-\vb*{x}\|)&1
\end{pmatrix} .
    \label{eq:Aelements}
\end{equation}
Each $2\times 2$ block has eigenvalues of the form $1\pm \sinc(\sqrt{N/b_0} \|\vb*{x'}-\vb*{x}\|)$ (see Ref.~\cite{skip_goetschy} for an analogous approach). If the blocks were independent, the eigenvalues of $S$ would take the form $\lambda=1+\xi$, where $\xi=c_++c_-$ is the sum of two independent random variables $c_+$ and $c_-$ uniformly distributed in $[0,a]$ and $[-a,0]$ respectively, with $a$ proportional to $\sqrt{b_0}$. The corresponding distribution would be a normalized symmetric triangular density $a^{-1} p_{\Delta}((x-1)/a)$, centred at $\lambda=1$ and with base $[1-a,1+a]$, where
\begin{equation}
p_{\Delta}(x)=
\begin{dcases} 1-|x| & \text{if } x \in [-1,1] \\
0 &  \text{otherwise} 
\,,
\end{dcases}
 \label{eq:triangle0}
\end{equation}	
namely,
\begin{equation}
\frac{1}{a} p_{\Delta}\Bigl(\frac{x-1}{a}\Bigr)=
\begin{dcases} \frac{a-\abs{x-1}}{a^2} & \text{if } x \in [1-a,1+a] \\
0 &  \text{otherwise} 
\,.
\end{dcases}
 \label{eq:triangle1}
\end{equation}	

\begin{figure}
\centering
\includegraphics[width=.7\columnwidth]{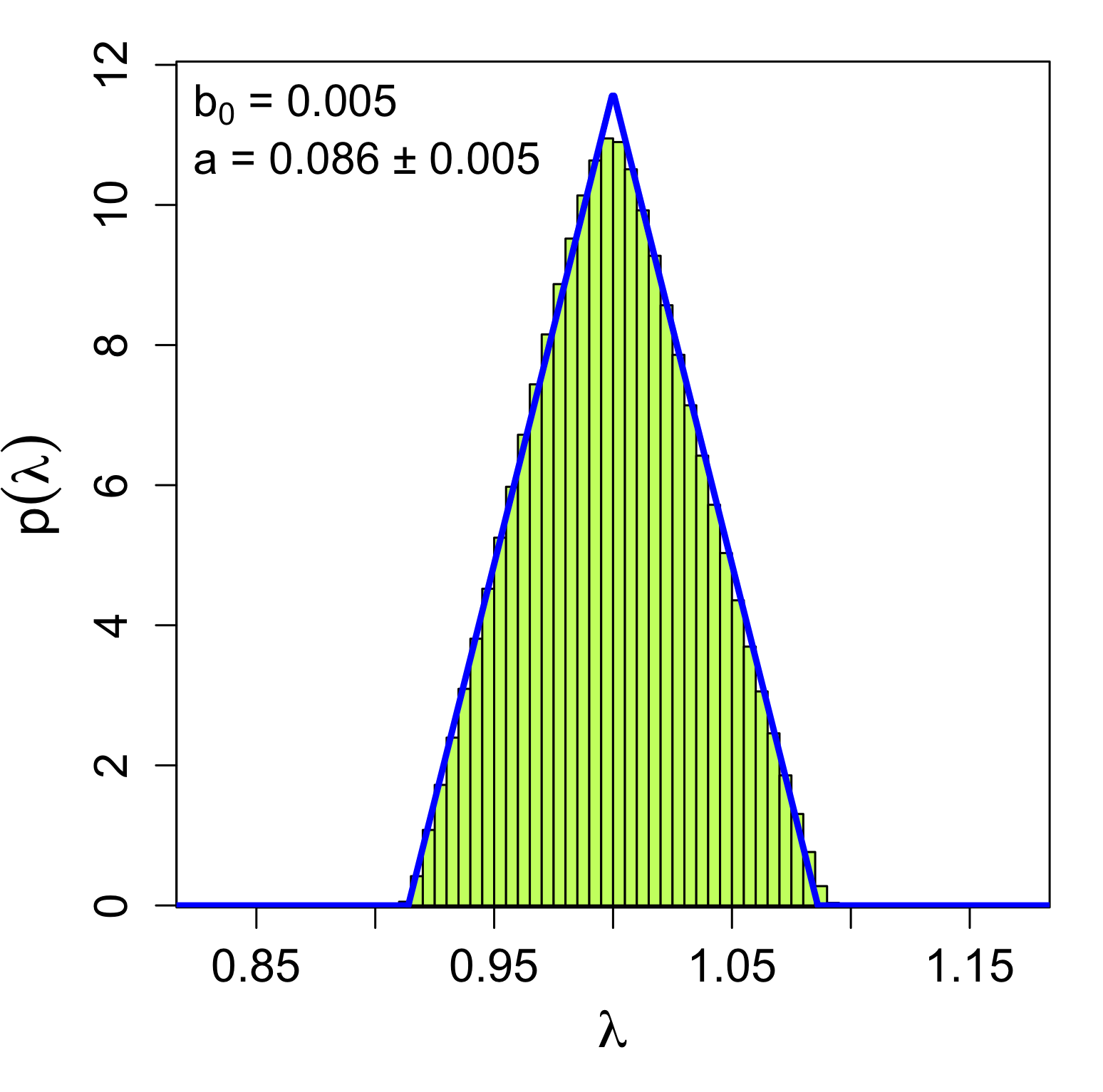}\\
\includegraphics[width=.7\columnwidth]{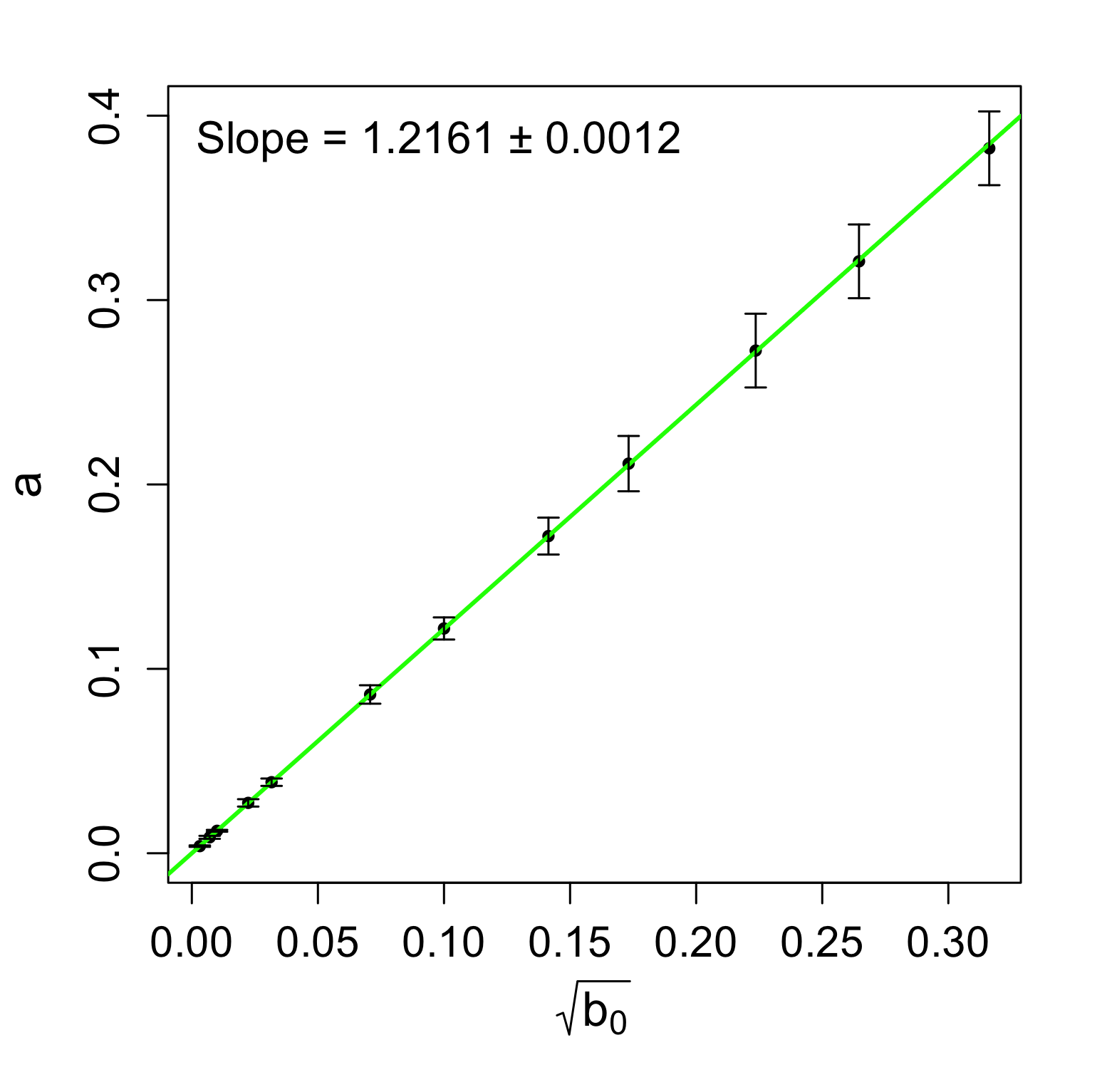}
\caption{Top: Histogram of the normalized eigenvalue distribution $p(\lambda)$ obtained from 20 samplings of the matrix $S$ in Eq.\ \eqref{eq:Sb0}, with $N=10000$ and $b_0=0.005$ (notice that there are few eigenvalues outside the range of the horizontal axis that have not been plotted). 
The histogram has been fitted with the one-parameter family of triangular distributions~\eqref{eq:triangle1}, yielding best-fit parameter $a=0.086 \pm 0.005$ (blue line). The error of $a$ has been set equal to the bin size of the histogram. Bottom: Parameter $a$ of the triangular distribution~\eqref{eq:triangle1} as a function of $\sqrt{b_0}$. 
}
\label{fig:trianglefit}
\end{figure}
The value of $a$ can be chosen by imposing the matching condition of the second moment of~\eqref{eq:triangle1} ,
\begin{equation}	
\int \frac{x^2}{a} p_{\Delta}\Bigl(\frac{x-1}{a}\Bigr)dx=1+\frac{a^2}{6},
\end{equation}	
with the limiting second moments~\eqref{eq:2ndmomlim} of $p(\lambda)$. This gives
\begin{equation}
a=\sqrt{\frac{3b_0}{2}}. 
    \label{eq:a_estimate}
\end{equation}
Equation~\eqref{eq:a_estimate}  relates the support of the eigenvalue density and $b_0$, for $b_0\ll1$. 
We checked numerically that indeed $a$ scales linearly with  $\sqrt{b_0}$, 
for small values of $b_0$. See  Fig.~\ref{fig:trianglefit}. The (weighted) linear fit yields an offset $(2.096 \pm 1.672) \times 10^{-5}$, compatible with zero, and a slope $1.2161 \pm 0.0012$, to be compared with the theoretical value $\sqrt{3/2} \approx 1.2247$ in Eq.~\eqref{eq:a_estimate}.

The agreement between the triangular density in Eq.~\eqref{eq:triangle1} and $p(\lambda)$ for $b_0 \ll 1$ is surprisingly good, given the simple assumptions underlying the derivation. However, it is possible  to show that the triangular density is only an approximation, and 
\begin{equation}
\label{eq:plambdab0_v}
\lim_{b_0\to0}\lim_{N \to\infty}p(\lambda)\neq   \frac{1}{a} p_{\Delta}\Bigl(\frac{\lambda-1}{a}\Bigr).
\end{equation}
In order to simplify the analysis that follows, it is convenient to centre and rescale the matrix $S$ in~\eqref{eq:Sb0} as
\begin{equation}
Q= \sqrt{\frac{2}{3b_0}}(S-I) .
    \label{eq:Q}
\end{equation}
Notice that the matrix $Q$ has vanishing elements on the diagonal.
 We should compare the eigenvalue density of $Q$ for $N \to + \infty$ and small values of $b_0$ with the centred triangular density $p_{\Delta}(x)$ of Eq.~\eqref{eq:triangle0}. We have that
 \begin{equation}
 \lim_{b_0\to0}\lim_{N\to\infty}\frac{1}{N}\langle\Tr Q^m\rangle= \int \lambda^m p_{\Delta}(\lambda),
 \end{equation}
 for $m=0,1,2$. We also verified the identity for $m=3$. We now show that the matching condition \emph{fails} for the fourth moment, $m=4$. The fourth moment of the triangular distribution is \begin{equation}
\int  \lambda^4 \, p_{\Delta}(\lambda)d\lambda=\frac{1}{15}= 0.0667\ldots.
    \label{eq:triangle_4mom}
\end{equation}

On the other hand, the fourth moment of $Q$ is 
\begin{align}
\begin{aligned}
\frac{1}{N}\langle \Tr Q^4\rangle&=\frac{1}{N}\sum_{ijkl }\langle Q_{ij}Q_{jk}Q_{kl}Q_{li}\rangle \\
&=\frac{1}{N}\biggl(\frac{2}{3b_0}\biggr)^2 {\sum_{ijkl}}'  \langle S_{ij}S_{jk}S_{kl}S_{li}\rangle\, ,
    \label{eq:4mom}
    \end{aligned}
\end{align}
where$\sum'$ denotes a sum over distinct consecutive indices $i\neq j\neq k\neq l\neq i$. There are three possibilities for the indices $i,j,k,l$:
\begin{enumerate}
    \item there are two pairs of equal indices (i.e.\ $i=k$ and $j=l$);
    \item there is one pair of equal indices (this contribution has multiplicity 2);
    \item all indices are different.
\end{enumerate}
Therefore, taking into account the symmetries of $S$, Eq.~\eqref{eq:4mom} can be written as the sum of  three contributions,
\begin{align}
\begin{aligned}
\frac{\langle \Tr Q^4 \rangle}{N}&=\frac{4}{9b_0^2 N} \biggl[N(N-1)\langle S_{ij}^4\rangle \\
&+2N(N-1)(N-2)\langle S_{ij}^2 S_{il}^2 \rangle \\
&+N(N-1)(N-2)(N-3)\langle S_{ij} S_{jk}S_{kl}S_{li} \rangle\biggr]
    \label{eq:4momQ}\,.
\end{aligned}
\end{align}

The expectation values in Eq.~\eqref{eq:4momQ} are difficult to compute for fixed $N$ and $b_0$, but the calculations become simpler in the large-$N$ limit. Indeed, from~\eqref{eq:density_R} we have
\begin{equation}
\langle S_{ij}^4 \rangle=O((b_0/N)^{3/2})
    \label{eq:Es^4}\,,
\end{equation}
so that the first term in Eq.~\eqref{eq:4momQ} vanishes asymptotically. Furthermore, it can be shown numerically that $\langle S_{ij}S_{jk}S_{kl}S_{li} \rangle$ scales like $b_0^3/N^3$, so that the third term in Eq.~\eqref{eq:4momQ} in the large-$N$ limit scales as $b_0$, and hence vanishes in the case $b_0 \to 0$ under consideration.  As a consequence, Eq.~\eqref{eq:4momQ} reduces to
\begin{equation}
\frac{1}{N}\langle \Tr Q^4\rangle\sim \frac{8N^2}{9b_0^2}\langle S_{ij}^2 S_{il}^2 \rangle
    \label{eq:4momQlimit}\,.
\end{equation}
In order to compute the above expectation value, we use  the joint probability density function~\eqref{eq:density_RR'}
and we get
\begin{align}
\lim_{b_0\to0}\lim_{N\to\infty}\frac{1}{N}\langle \Tr Q^4\rangle&\sim\frac{8}{9}\,\frac{2}{\pi\sqrt{3}}\,\sum_{j=0}^{+\infty}\frac{1}{(2j+1)!\,3^{2j+1}}\nonumber\\&\times\left( \frac{1}{2}\int_{0}^{+\infty}r^{2j}\,e^{-r^2/3}dr\right)^2\nonumber\\
    &=\frac{\pi}{27\sqrt{3}}= 0.0672\ldots.
    \label{eq:4mom_result}
    \end{align}
This value should be compared with the fourth moment~\eqref{eq:triangle_4mom} of the triangular density. The discrepancy between Eqs.~\eqref{eq:4mom_result} and~\eqref{eq:triangle_4mom} is quite small, approximately $0.7\%$, corroborating (and correcting at the same time) the triangular approximation for $N\to +\infty$ and $b_0 \to 0$. Considering the simplicity of the triangular density, the approximation is useful as a quick benchmark against numerical data.

\section{Microscopic statistics}

\label{sec:NNSD}

In the statistical approach to complex deterministic quantum systems, the guiding idea is that some central quantum features (energy levels, decay rates, etc.), are so erratic and irregular that their precise values are inconsequential. It is only their statistical properties that hold significance~\cite{Wigner2}. As matrices are intrinsic to quantum mechanics, random matrices are of paramount importance in the application of statistics to quantum problems.

In the field of quantum chaos it was soon realized that the statistics of the
microscopic behaviour of the energy levels of a quantum system could help discriminating between systems whose classical counterpart is integrable or chaotic. 

For generic integrable models, the energy levels follow the distribution of eigenvalues of diagonal matrices with independent identically distributed (i.i.d.) diagonal entries. This implies that their correlation functions, after unfolding, coincide with the ones of a Poisson process.  This is the celebrated Berry-Tabor conjecture \cite{congettura,proof_congettura}. 

On the other hand, according to the Bohigas-Giannoni-Schmit conjecture \cite{BGS}, it is expected that the energy level statistics of generic chaotic systems follow the eigenvalues of standard random matrix ensembles (full, with independent entries) dependent only on system symmetries, whose correlation functions are known explicitly.

 However, different random matrix models have
different  eigenvalue densities, and a meaningful comparison between microscopic statistics requires a transformation called \emph{unfolding}~\cite{mehta}. The unfolded eigenvalues ${\lambda^{\mathrm{unf}}_i}$
 and the true levels $\lambda_i$ are related via ${\lambda^{\mathrm{unf}}_i}=\mathcal{N}(\lambda_i)$, where $\mathcal{N}(x)=\int^x\rho(y)dy$ is the mean number of eigenvalues less than $x$. The unfolded spectrum has  mean level spacing equal to $1$.

\subsection{Nearest-neighbor spacing distribution}

The universality classes are characterized by the $N\to\infty$ behaviour of microscopic statistics, such as the nearest-neighbor spacing density (NNSD) $p(s)$, which represents the probability density of two eigenvalues being separated by a distance $s$ and no other eigenvalues in between,
\begin{equation}
 s_i=\lambda^{\mathrm{unf}}_{i+1}-\lambda^{\mathrm{unf}}_i,
    \label{eq:spacing}
\end{equation}
where $\lambda^{\mathrm{unf}}_1 \leq \lambda^{\mathrm{unf}}_2 \leq \dots \leq \lambda^{\mathrm{unf}}_N$ (the unfolding preserves the order of the eigenvalues).
For Poisson statistics, 
\begin{equation}
 p_{\mathrm{P}}(s)=e^{-s}
    \label{eq:Poisson}\,,
\end{equation}
which is the asymptotic distribution of the spacings between i.i.d.\ random variables with unit mean spacing. 
 Conversely, standard real random matrices display \emph{level repulsion} and the NNSD is described by the Wigner surmise
 \begin{equation}
p_{\mathrm{W}}(s)=\frac{\pi s}{2}e^{-\pi s^2/4}
    \label{eq:surmise}\,,
\end{equation}
also known as Wigner-Dyson distribution.

The Berry-Tabor and  the Bohigas-Giannoni-Schmit conjectures form a cornerstone of quantum chaos and have been effectively employed in various problems ranging from nuclear physics to number theory. 
However, they do not account for all possible types of models.
Several physical systems show a behavior that deviates from the dichotomy `integrable vs chaotic'. 

Since $S$ is a Euclidean random matrix, it is not clear \emph{a priori} whether it falls into one of the standard random matrix universality classes. In fact, we show below that at least locally, in some regions of the spectrum, the NNSD can be described by a suitable two-parameter family of \emph{Wigner-like distributions},
\begin{align}
 p(q,r;s)=a s^q \, e^{-b s^{r}},\quad q>-1, r>0.
    \label{eq:2par}
    \end{align}
The constants $a$ and $b$ are determined from the normalization conditions
\begin{equation}
\label{eq:norm_ab}
\int_{0}^{\infty} p(q,r;s)ds=1,\quad \int_{0}^{\infty} s p(q,r;s)ds=1,
\end{equation}
namely $a=r \left[\Gamma \left(\frac{q+2}{r}\right)\right]^{q+1}/\left[\Gamma \left(\frac{q+1}{r}\right)\right]^{q+2}$, and $b=\left[\Gamma \left(\frac{q+2}{r}\right)/\Gamma \left(\frac{q+1}{r}\right)\right]^r$. 

Notice that the exponent $q$ accounts for the level repulsion, while the exponent $r$ governs the tail for large spacings.
In particular, if $q=0$ and $r=1$, Eq.~\eqref{eq:2par} coincides with the Poisson distribution~\eqref{eq:Poisson}, while for $q=1$ and $r=2$ it reduces to the Wigner-Dyson distribution~\eqref{eq:surmise}.

For $r=q+1$ one obtaines the one-parameter family of Brody distributions~\cite{brody_1973} interpolating between Wigner-Dyson and Poisson, and originally proposed to describe systems characterized by a transition between the integrable and chaotic regime and by the coexistence of regions with regular and ergodic motion~\cite{brody,transition,transition2}. The parameter $q$ can be interpreted as a measure of the degree of chaoticity of the system under consideration.

Setting $r=1$ in~\eqref{eq:2par}, we get another widely used one-parameter family of spacing distribution, associated to quantum systems whose classical counterpart is pseudo-integrable. In those cases, the NNSD~\eqref{eq:2par} is the spacing distribution of the so-called intermediate statistics~\cite{intermediate,intermediate2,toep},
proposed to describe quantum systems with multifractal wavefunctions, i.e.\ with an intermediate behavior between localized and fully extended eigenstates.

\subsection{Eigenvector statistics} 
\label{sec:PR}

The microscopic spectral statistics can be associated with the support of the corresponding eigenvectors
\begin{equation}
    \ket{\Psi} = \sum_{j=1}^N \Psi_j \ket{j} .
\end{equation}
 More precisely, corresponding to Poisson statistics, one expect the eigenvectors to be \emph{localized}. Indeed, a random  matrix diagonal in a basis $\{\ket{j}\}$ with i.i.d.\ diagonal entries yields  $|\Psi_j|^2=\delta_{j\ell}$ for some $\ell$.  On the other hand, standard real random matrices, described by a Wigner-Dyson NNSD \eqref{eq:surmise}, have generically \emph{delocalized} eigenvectors~\cite{loc-deloc,loc-deloc2}. For example, a generic eigenvector of a Gaussian orthogonal matrix is uniformly distributed on the $N$-sphere, and hence its coefficients yield, for large-$N$, the so-called \emph{Porter-Thomas} (PT) distribution~\cite{porter_thomas}
 \begin{equation}
p_{\mathrm{PT}}(u)=\frac{1}{\sqrt{2\pi}}e^{-u^2/2},
    \label{eq:PT}
\end{equation}
where $u=\sqrt{N}\Psi_j$ with $\Psi_j$ being the $j$-th (real) component of $\ket{\Psi}$. We choose a distribution of real amplitudes since the eigenvectors of $S$, which is real and symmetric, can always be chosen real.
Notice that the vectors $\ket{\Psi}$ are normalized  on average, and the deviations from normalization for a given sample become irrelevant in the large-$N$ regime.

A widespread indicator of the (de)localization of the normalized  eigenvectors is the \emph{participation ratio} (PR),
\begin{equation}
\Pi(\ket{\Psi}) =  \frac{1}{\sum_{j=1}^N\abs{\Psi_j}^{4}}.
    \label{eq:PR}
\end{equation}
It is related to the effective number of nonvanishing components of the vector $\ket{\Psi}$. In particular, a localized eigenvector with $\abs{\Psi_j}^2=\delta_{jk}$ for some $k$ has  participation ratio $1$, while for a uniform eigenvector, such that $\abs{\Psi_j}^2=1/N$ for all $j$, $\Pi(\ket{\Psi})=N$ is maximum.

Finer information is encoded in the large-$N$ asymptotics of the \emph{moments} of the eigenvectors~\cite{toep,multifractal}:
\begin{equation}
M_q=  \Bigl\langle \sum_{j=1}^N \abs{\Psi_j}^{2q}  \Bigr\rangle \sim C_q N^{-\tau(q)}
    \label{eq:multifractal}, \quad \text{as } N\to\infty.
\end{equation}
Notice that $M_1 = 1$ due to the normalization of the eigenvectors, while the second moment $M_2$ coincides with the inverse of the participation ratio $\Pi$ in~\eqref{eq:PR}. 
The exponent $\tau(q)$ defines the \emph{fractal dimension}: $D_q= \frac{\tau(q)}{q-1}$. $D_q$ determines the fraction of nonzero components of the eigenvectors at a certain
scale.
For fully extended eigenvectors $D_q=1$ and for localized eigenvectors $D_q = 0$. The case when $D_q$ depends on $q$ corresponds to multifractal eigenvectors.

One can compute explicitly the large-$N$ scaling of the moments in Eq.~\eqref{eq:multifractal} under the assumption that a delocalized eigenvector follows the PT distribution~\eqref{eq:PT}:
\begin{align}
 \begin{aligned}
 M_q
 &\sim N \int_{-\infty}^{+\infty} \Psi_j^{2q}\sqrt{\frac{N}{2\pi}}e^{-\frac{N}{2}\Psi_j^2}d\Psi_j= \frac{(2q-1)!!}{N^{q-1}}
 \label{eq:M_q}\,.
 \end{aligned}   
\end{align}
This equation not only reproduces the expected scaling law for delocalized states $1/M_q \sim N^{q-1}$ discussed above (in this case $D_q=1$ for all $q$), but also provides the explicit value of the constants $C_q=(2q-1)!!$ that can be compared with the results of numerical simulations.
For instance, for the participation ratio we get
\begin{equation}
\Pi=\frac{1}{M_2}=\frac{N}{C_2}=\frac{N}{3}.
    \label{eq:C_2}
\end{equation}

\subsection{Numerical results}\label{sec:numerics} 
We studied the NNSD of the matrix $S$ defined in Eq.~\eqref{eq:Sb0}, in the limit $N \to \infty$ for different values of the cooperativeness parameter $b_0$.
The following analysis has been performed by sampling matrices $S$ of size $N=10^4$ for four representative values of the cooperativeness parameter: $b_0=0.1,1,3,10$.
The eigenvalues of $S$ have been used to compute the spacings $s_i \,\, \forall i=1, \dots, N-1$ according to Eq.~\eqref{eq:spacing}.

To study the local NNSD of the matrix $S$ with $b_0=0.1,1,3$ we considered several intervals of the spectrum containing 1000 eigenvalues. For $b_0=10$, instead, we restricted our attention to regions containing 600 eigenvalues, since the spectrum is much broader [as it can be seen comparing Fig.~\ref{fig:b10} with~\ref{fig:b0_1}, \ref{fig:b1} and~\ref{fig:b3}] and hence too large a fraction of the eigenvalues would mix different regimes. 
For each region of the spectrum, we compared the spacings of the unfolded eigenvalues with the two-parameter NNSD~\eqref{eq:2par}. An example is reported in Fig.\ \ref{fig:2parameterfit}, which shows the spacing distribution for $b_0=1$ in the region $\lambda \in [1.067,1.225]$.

\begin{figure}[hb]
\centering
\includegraphics[width=.6\columnwidth]{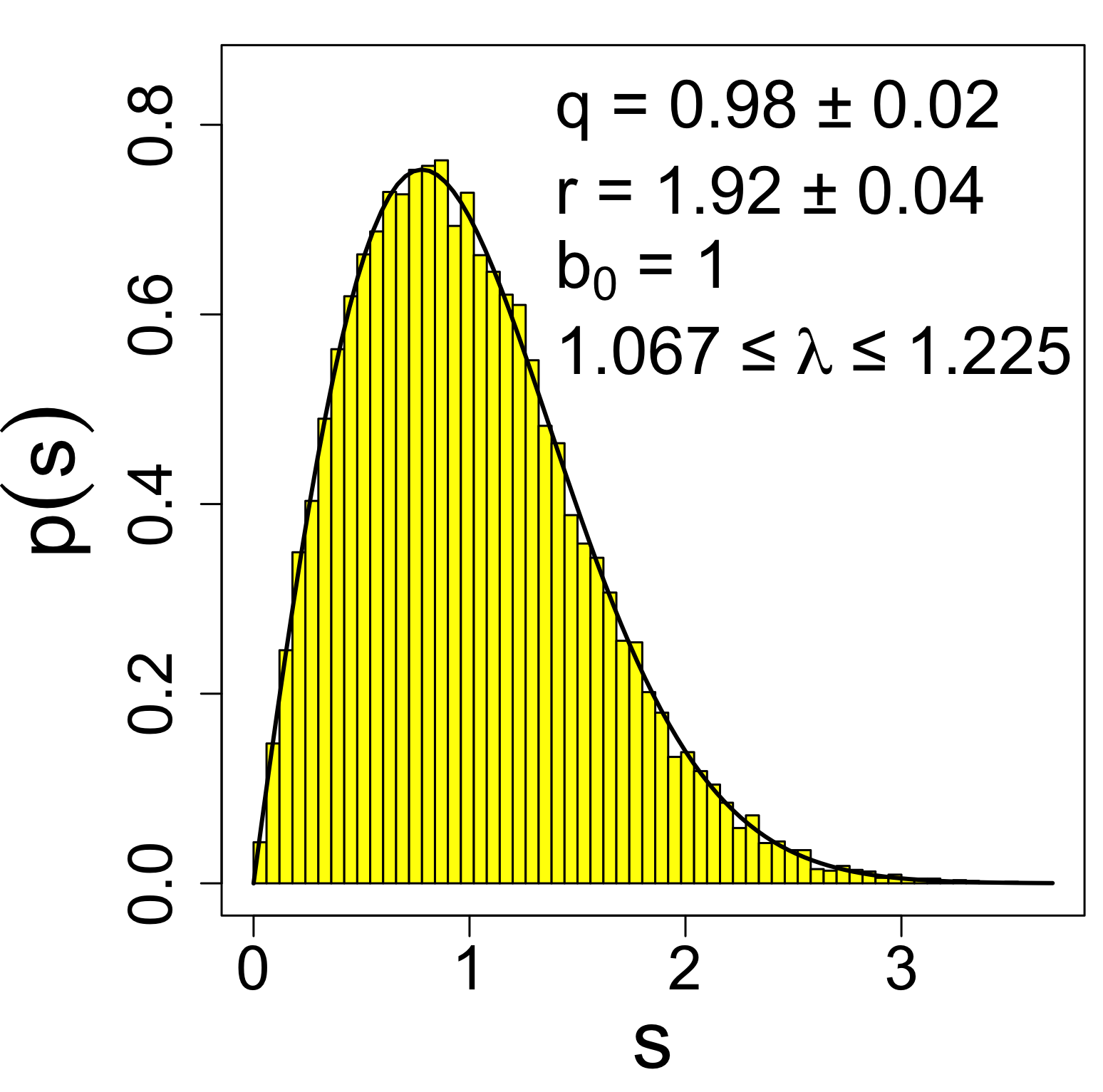}
\caption{Histogram of the spacing distribution in the region corresponding to eigenvalues $\lambda \in [1.067,1.225]$ for $b_0=1$, fitted with the two-parameter NNSD~\eqref{eq:2par}. Notice the best-fit values of the parameters $q\simeq 1$ and $r\simeq 2$, which correspond to Wigner-Dyson and hence to delocalized states. }
\label{fig:2parameterfit}
\end{figure}

Our findings are displayed in Figs. \ref{fig:eigenvalues}(e)-(h)
and can be summarized as follows. 
For each value of $b_0$, there is a central region of the spectrum around $\lambda=1$ (which corresponds to the decay rate of an isolated atom) whose NNSD is described by the Wigner-like surmise \eqref{eq:2par} with parameters $q \simeq 1$ and $r \simeq 2$, which corresponds to the Wigner-Dyson distribution~\eqref{eq:surmise}. Moving towards the edges of this region, we approach a subradiant and a superradiant region of the spectrum for which the spacing distribution is still described by~\eqref{eq:2par}, but the best-fit values of $q$ and $r$ deviate from the Wigner-Dyson case.
Finally, the NNSD at the tails of the spectrum (the extreme sub- and superradiant regions) is not described by~\eqref{eq:2par} (except for the extreme subradiant region for $b_0=3$). 
This is not surprising, since in general, the bulk and the tails of the spectrum of random matrices behave very differently. 

\begin{figure*}[t]
\centering
\subfigure[]{\includegraphics[trim={0 0 0 1cm}, clip,width=0.5\columnwidth]{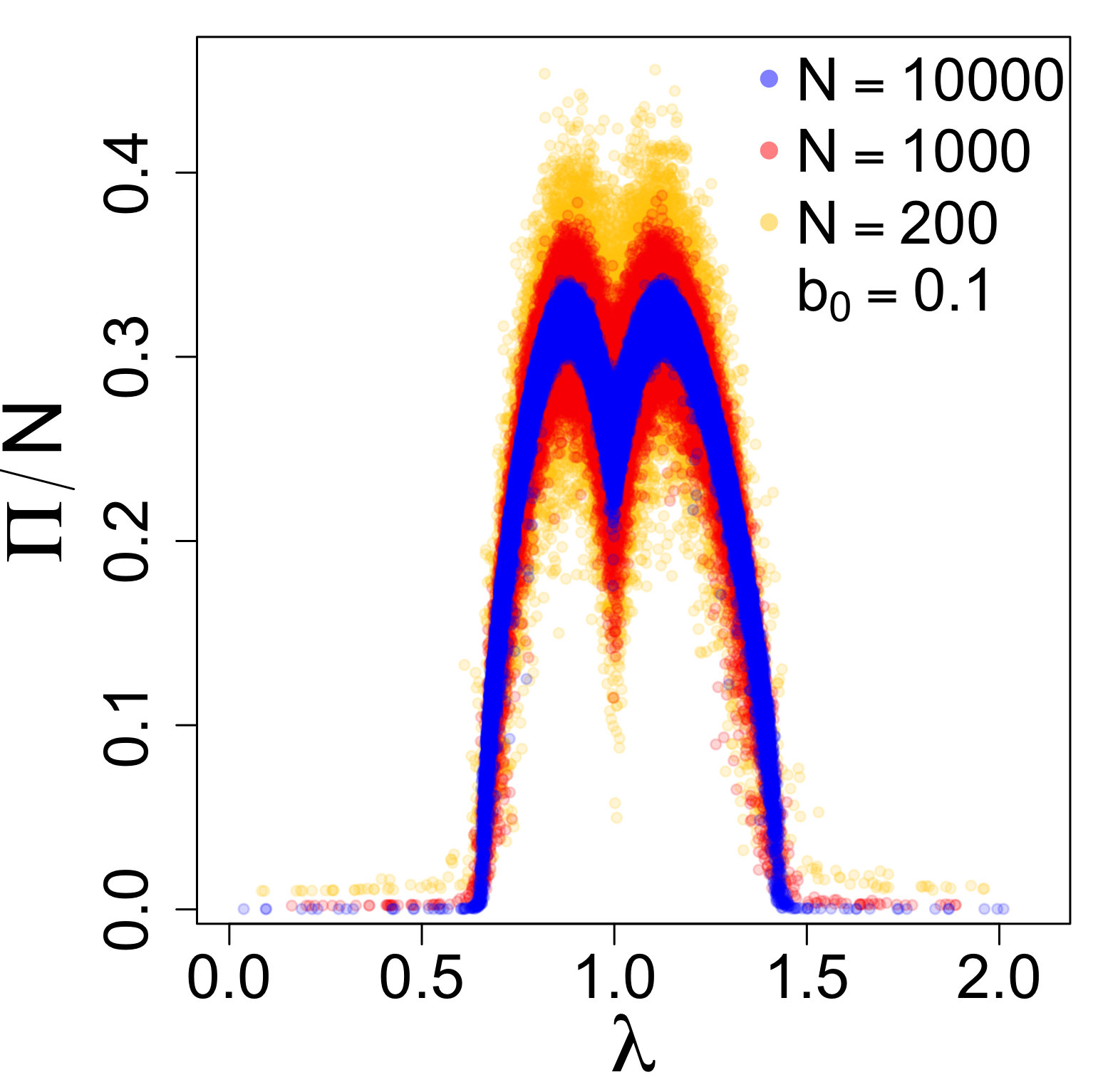}\label{fig:PRsovrappostib0_1}}
\subfigure[]{\includegraphics[trim={0 0 0 1cm}, clip,width=0.5\columnwidth]{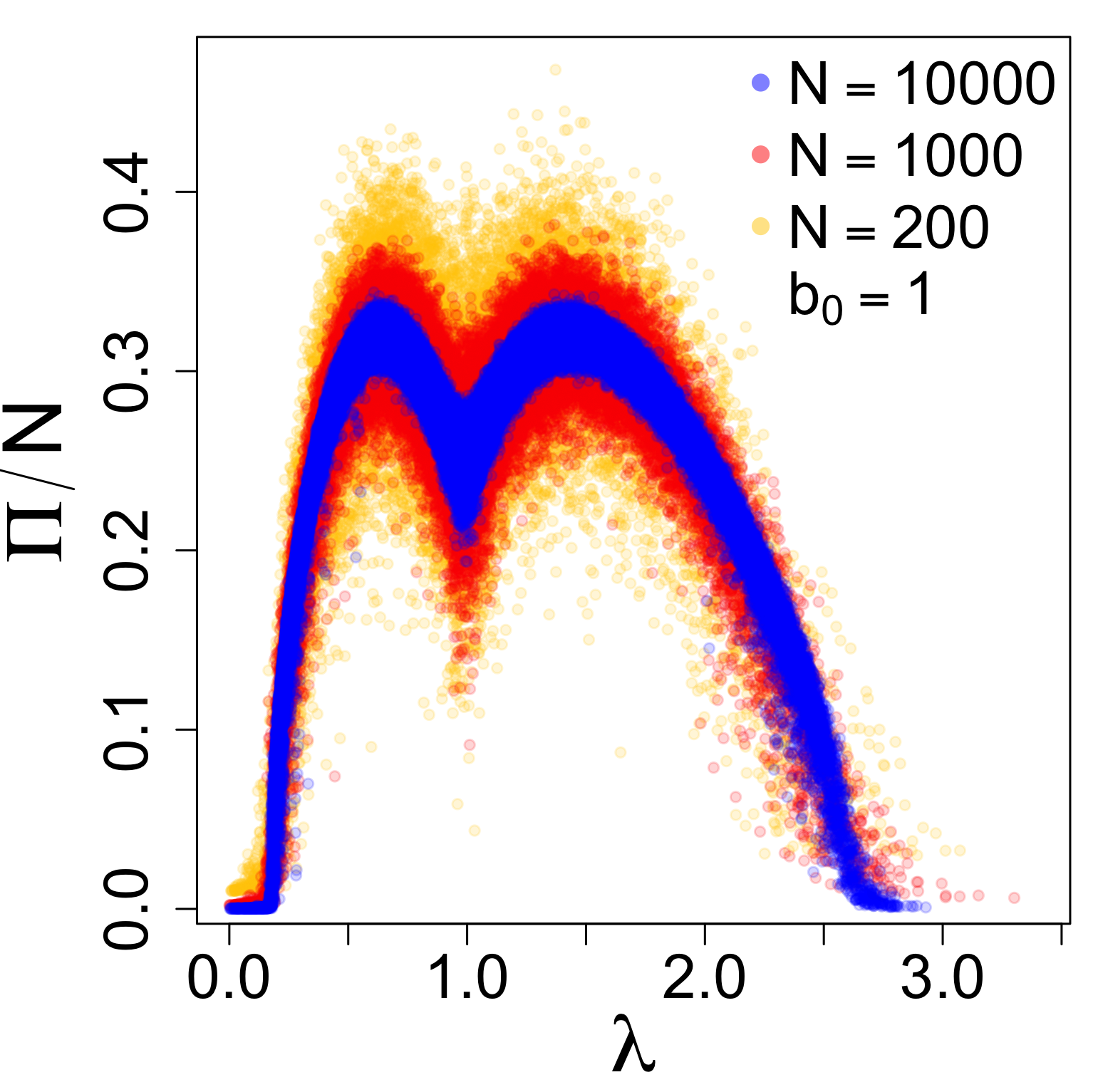}\label{fig:PRsovrappostib1}}
\subfigure[]{\includegraphics[trim={0 0 0 1cm}, clip, width=0.5\columnwidth]{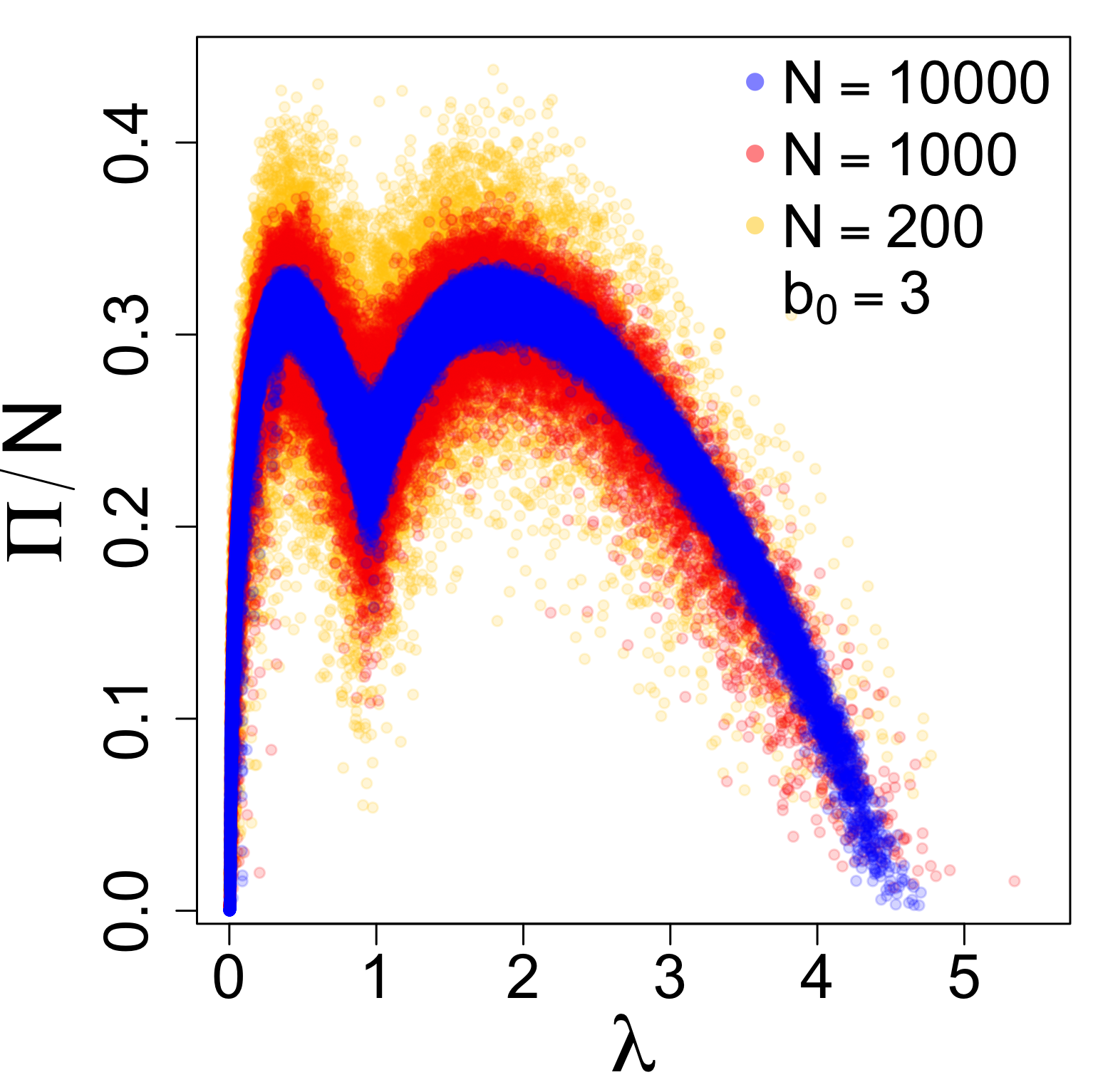}\label{fig:PRsovrappostib3}}
\subfigure[]{\includegraphics[trim={0 0 0 1cm}, clip,width=0.5\columnwidth]{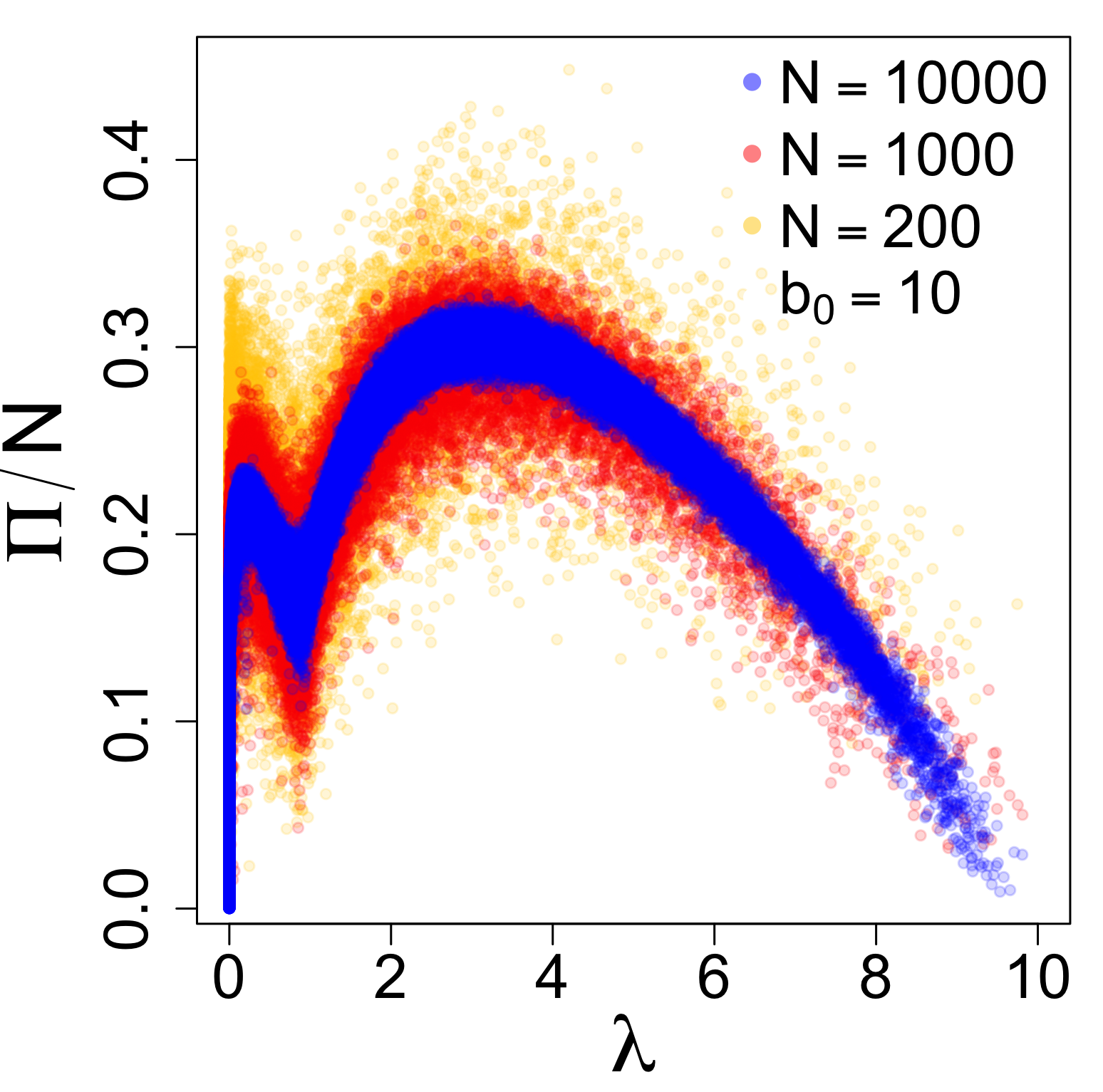}\label{fig:PRsovrappostib10}}
\caption{Participation ratio $\Pi$ divided by $N$ as a function of the corresponding eigenvalue for $b_0=0.1,1,3,10$. For each value of $b_0$ we show the plots of $\Pi/N$ corresponding to $N=200, 1000, 10000$. Notice that for all values of $b_0$ the participation ratio of the most subradiant state is (very) close to 2 and that in the central region of the spectrum the states are delocalized with a participation ratio $\simeq N/3$, in accord with Eq.\ \eqref{eq:C_2}. 
}
\label{fig:PRsovrapposti}
\end{figure*}

\begin{figure*}[t]
\centering
\subfigure[]{\includegraphics[trim={0 0 0 1cm}, clip,width=.7\columnwidth]{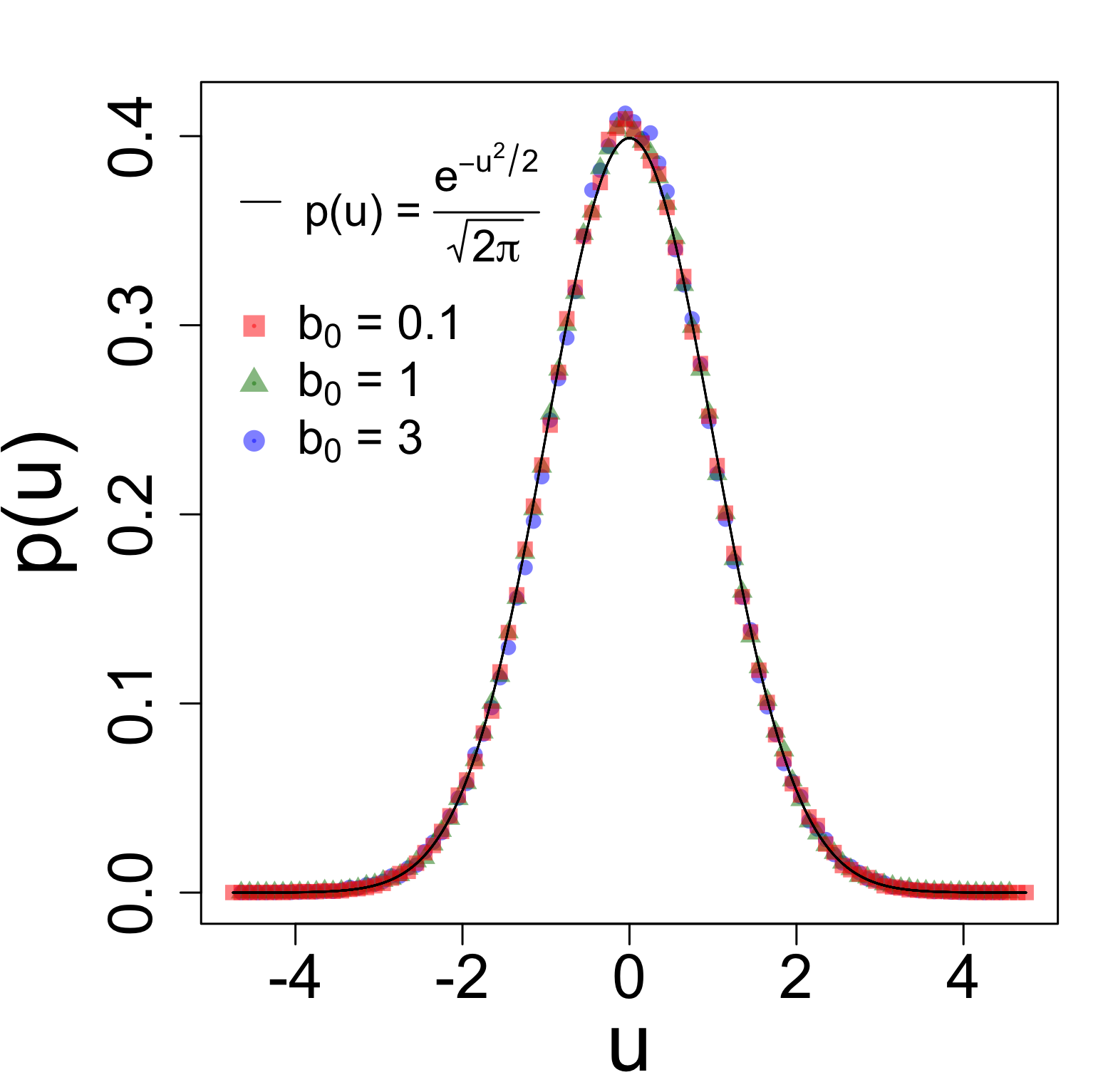}\label{fig:porter_picco_subradiante}}
\subfigure[]{\includegraphics[trim={0 0 0 1cm}, clip,width=.7\columnwidth]{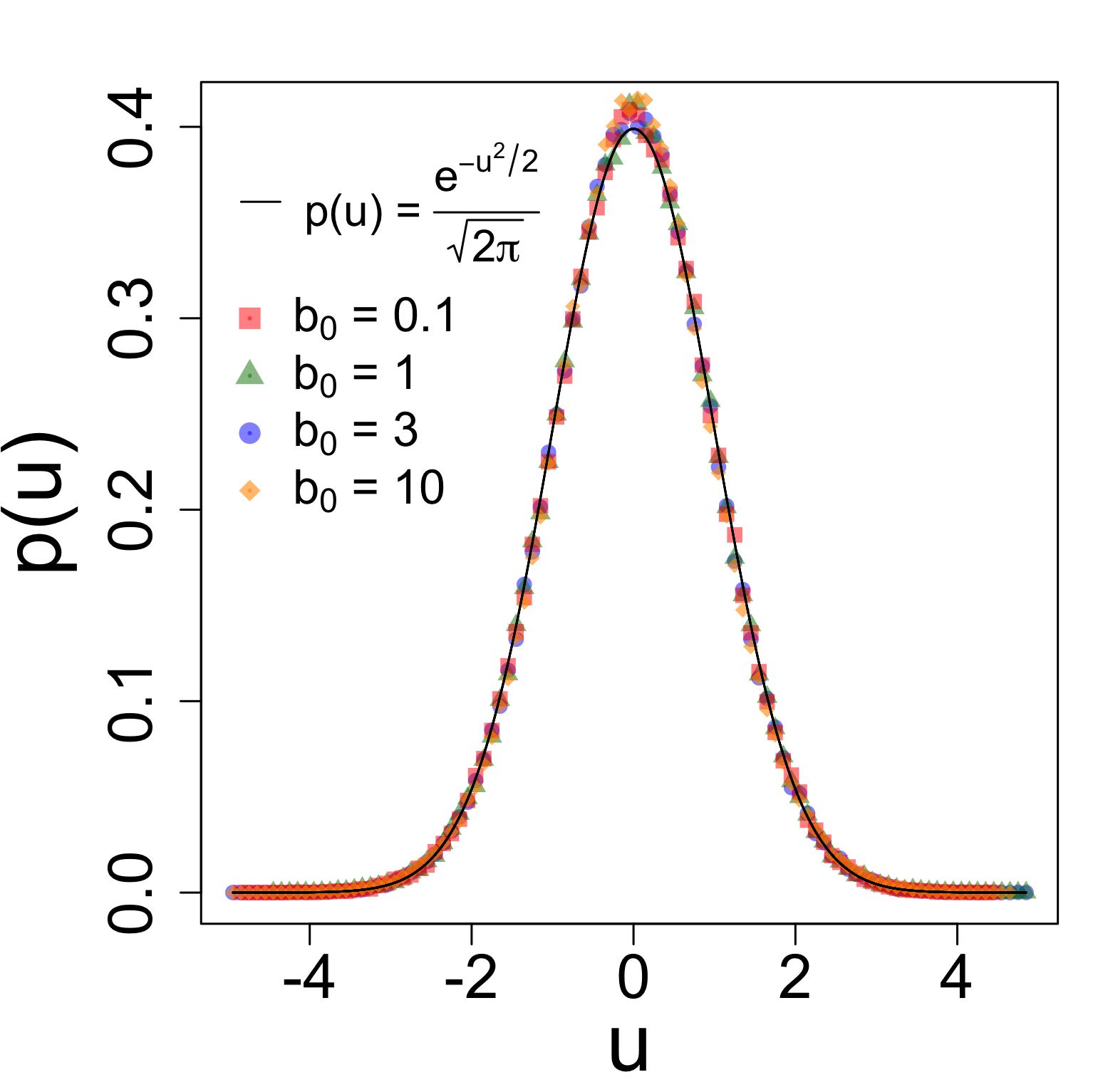}\label{fig:porter_picco_superradiante}}
\caption{
Distribution of $u=\sqrt{N}\Psi_j$, where $\Psi_j$ is the $j$-th component of the eigenvector $\ket{\Psi}$ of the matrix $S$ with $b_0=0.1,1,3,10$. The solid line is the the Porter-Thomas distribution \eqref{eq:PT},  expected for random delocalized eigenvectors.
(a) eigenvector $\ket{\Psi}$ corresponding to the subradiant eigenvalue for which the participation ratio $\Pi$ has a maximum (see Fig.\ \ref{fig:PRsovrapposti}). 
(b) eigenvector $\ket{\Psi}$ corresponding to the superradiant eigenvalue for which the participation ratio has a maximum (see Fig.\ \ref{fig:PRsovrapposti}).  For $b_0=10$ the PT distribution is no longer valid in the subradiant part, and is not displayed in the figure. 
}
\label{fig:111}
\end{figure*}

The numerical analysis shows that the full spacing distribution of $S$ in the bulk of the spectrum is difficult to characterize, even if one considers a two-parameter Wigner-like surmise~\eqref{eq:2par}. Rather than trying to fit the full distribution $p(s)$ we therefore set ourselves the task of verifying whether at least the level repulsion  is a generic property in the bulk. This amounts to analyze the small $s\to0$ behaviour of $p(s)$. To this purpose, we have studied the cumulative probability distribution for small spacings, i.e.\ $P(s \leq s_0)$ as $s_0\to0$.
Therefore, for several values of $b_0$ and different subregion of the bulk of the spectrum, we computed $N(s_0)=P(s \leq s_0)$ as a function of $s_0$. We found numerically (data not shown) that  $N(s_0)\sim s_0^{q+1}$, with $q\simeq1$ so that the NNSD in the bulk vanishes at small spacings as $p(s)\sim s$.
According to the statistical approach to quantum chaos (see Sec.\ \ref{sec:NNSD}), these findings suggest that the bulk of the spectrum of the ERM $S$ displays level repulsion and is associated with quantum phenomena whose classical counterparts are chaotic.

As explained in Sec.\ \ref{sec:PR}, since we found that the NNSD in the bulk of the spectrum is described by the Wigner-Dyson distribution \eqref{eq:surmise}, we expect the corresponding eigenvectors to be delocalized. Therefore, to complete this analysis we looked at the statistics of the eigenvectors in the bulk of the spectrum. First, we computed the PR of the eigenvectors of $S$ for the same representative values of the cooperativeness parameter $b_0=0.1,1,3,10$, and  $N=200, 1000, 10000$. See Fig.~\ref{fig:PRsovrapposti}. We found that the average PR scales with $N$, as expected for delocalized states. In particular, the plots in Fig.~\ref{fig:PRsovrapposti} show that the average of $\Pi/N$ converges for large $N$, and we verified that the fluctuations of $\Pi/N$ are of order $N^{-1/2}$. 
We note that the average PR has two maxima, corresponding to eigenvalues in the subradiant and superradiant regions of the spectrum, respectively, and that $\Pi$ in the central region of the spectrum is of the order of $N/3$, as predicted by the PT distribution, see Eq.~\eqref{eq:C_2}. A further evidence of the delocalization of eigenvectors corresponding to the bulk of the spectrum comes from Fig.\ \ref{fig:111}, which shows the distribution of $u=\sqrt{N}\Psi_j$, where $\Psi_j$ is the $j$-th component of the eigenvector $\ket{\Psi}$ of $S$ corresponding to the sub- and superradiant eigenvalues for which the PR is maximum. As can be seen, the distribution $p(u)$ is Porter-Thomas \eqref{eq:PT}, expected for random delocalized eigenvectors.

Moreover, we checked that for all values of $b_0$, the participation ratio of the most subradiant states is (very) close to $2$, indicating that the lowest emission rates are associated with excitations shared between a pair of atoms, with antisymmetric excitation amplitudes (see the limiting case~\eqref{eq:psim}). Let us note that the participation ratio becomes more asymmetric as $b_0$ grows (see Fig.\ \ref{fig:PRsovrapposti}), as expected from the increasing asymmetry of the spectrum. Moreover, for all values of $b_0$ the PR drops near $\lambda=1$. We offer no explanation for this behaviour. A similar behaviour was also observed in Ref.~\cite{guerin2017population} for a related matrix.

\begin{figure}
\centering
\includegraphics[scale=0.11]{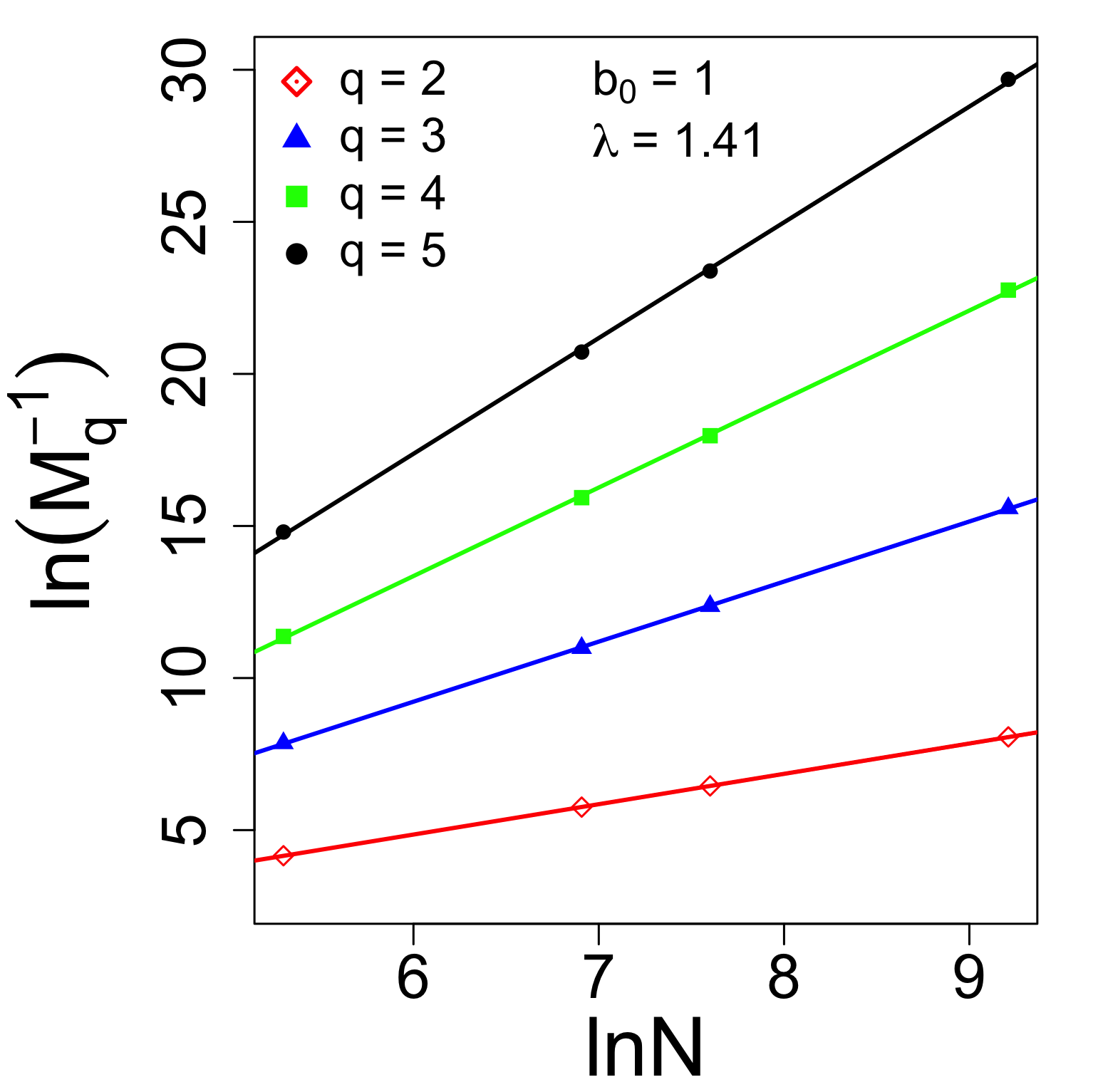}
\caption{Log-log plot of the mean inverse moment $1/M_q$ with $q=2,3,4,5$, as a function of $N$ in the region near the superradiant maximum $\lambda=1.41$ for $b_0=1$. For the second moment $q=2$ we obtain a best-fit line (red) with slope $\tau(2)=0.999 \pm 0.002$, indicating that the participation ratio scales linearly with $N$ as expected for delocalized states. Similarly, the blue, green and black lines corresponding with the third, fourth and fifth moment have best-fit slope $\tau(3)=1.98 \pm 0.01$, $\tau(4)=2.91 \pm 0.03$ and $\tau(5)=3.81 \pm 0.05$ respectively. These results corroborate the expected scaling $1/M_q \sim N^{q-1}$ for random delocalized states. See Eq.~\eqref{eq:M_q}. 
}
\label{fig:linearfit}
\end{figure}

We also numerically computed  the higher moments~\eqref{eq:multifractal} for the eigenvectors in the central part of the spectrum, and extracted the value of the fractal dimension $D_q$. See Fig.~\ref{fig:linearfit}. We found:
\begin{align}\label{eq:fractaldim_estimate}
\begin{aligned}
&D_2 \simeq 0.998 \,, \quad
D_3 \simeq 0.985 \,, \\
&D_4 \simeq 0.967 \,, \quad
D_5 \simeq 0.950 \,.
\end{aligned}
\end{align}
The deviation from the maximum value 1 increases very slightly for increasing values of $q$, practically ruling out multifractality.

All these results lead to the conclusion that the eigenvectors of the matrix $S$ corresponding to the bulk of the spectrum are delocalized, as already suggested by the study of the NNSD. Hence the bulk of $S$ shows the typical features of chaotic quantum systems which can be successfully described using the universal statistical properties of classical random matrices.

\section{Conclusions}
\label{sec:conclusions}

We characterized the bulk spectral properties of the decay rate matrix $S$, defined in~\eqref{eq:S2}, related to the existence of subradiant and superradiant decay modes of a random cold atomic cloud. We identified a precise low-density scaling~\eqref{eq:Sb0} in which $S$ has a limit eigenvalue density parametrized by the cooperativeness parameter $b_0$. It is worth remarking that such a parameter has an intuitive physical relevance, as the number of atoms that coherently cooperate in photon emission. We corroborate this intuitive view with analytic and (strong) numerical evidences that $b_0$, regardless of its specific value, is the \textit{only} relevant parameter that determines the eigenvalue distribution in the limit of large number of atoms. For small values of $b_0$ we found that the asymptotic eigenvalue density  can be approximated by the triangular density in Eqs.\ \eqref{eq:triangle1} and \eqref{eq:a_estimate}. Then we studied the nearest-neighbour spacing distribution of the eigenvalues for several values of $b_0$ using a two-parameter family of Wigner-like distributions~\eqref{eq:2par}. We found that, although $S$ is a Euclidean random matrix, in the bulk of the spectrum of $S$ there is level repulsion. The eigenvector statistics (participation ratio and higher moments) confirms that the eigenvectors in the bulk are delocalized.

We conclude with some possible directions worth exploring and further food for thought. For the physical implications of the considered model, it would be interesting to study the microscopic statistics at the edges of the spectrum of $S$, which are physically related to strong subradiance and superradiance. Preliminary numerical calculations seem to indicate that the edges of $S$ do not fall into the KPZ universality class~\cite{KPZ,KPZ2,KPZ3} of standard random matrices characterized by the Tracy-Widom distribution. A relevant point, beyond the scope of the present study, would be to investigate the interplay between the purely dissipative dynamics, described by $S$, and the full Hamiltonian dynamics, entailed by the coherent energy shifts and transitions induced by coupling with the field. 
It should also be possible to compute explicitly the limiting moments of $S$ for all values of $b_0$, thus extending what we have presented here for the approximate triangular density at $b_0\to0$. Moreover, it could be useful to study similar cooperative effects and the associated ERMs in dimension $d\neq3$ \cite{lonigro2021stationary,Bordenave13}, in view of alternative experimental implementations.

\section*{Acknowledgments}

We acknowledge the support by Regione Puglia and QuantERA ERA-NET Cofund in Quantum Technologies (GA No.\ 731473), project PACE-IN, and by INFN through the project `QUANTUM'. FDC acknowledges the support by the Italian National Group of Mathematical Physics (GNFM- INdAM), and by Regione Puglia through the project `Research for Innovation' - UNIBA024. PF acknowledges the support by PNRR MUR project CN00000013-`Italian National Centre on HPC, Big Data and Quantum Computing'. SP and FVP acknowledge the support by PNRR MUR project PE0000023-NQSTI.
RB acknowledges the support from grants 2018/15554-5 and 2019/13143-0, S\~ao Paulo Research Foundation (FAPESP), from the Brazilian National Council for Scientific and Technological Development (CNPq, Grant No. 313886/2020-2), and from the French government, through the UCA J.E.D.I. Investments in the Future project managed by the National Research Agency
(ANR) with the reference number ANR-15-IDEX-01. RB and RK received support from the project STIC-AmSud (Ph879-17/CAPES 88887.521971/2020-00) and CAPES-COFECUB (CAPES 88887.711967/2022-00).

\appendix

\section{Positive-definiteness of the matrix $S$}
\label{app:positive}
The  $N\times N$ matrix $S$ with entries
\begin{equation}
S_{ij}=
\operatorname{sinc}\left(k_a \|\vb*{r}_i-\vb*{r}_j\|\right), 
\end {equation}
is non-negative definite for every $\vb*{r}_1,\ldots,\vb*{r}_N$. Indeed, setting $f(\vb*{r}-\vb*{r}')=\operatorname{sinc}\left(k_a \|\vb*{r}-\vb*{r}'\|\right)$ we have
\begin{align*}
\langle \psi,S\psi\rangle&=\sum_{i,j=1}^{N}f\left(\vb*{r}_i-\vb*{r}_j\right)\overline{\psi_i}\psi_j\\
&=\sum_{i,j=1}^{N}\frac{1}{\left(2\pi\right)^{\frac{3}{2}}}\int_{\mathbb{R}^3}\hat{f}(\vb*{k})e^{i k\cdot\left(\vb*{r}_i-\vb*{r}_j\right)}\overline{\psi_i}\psi_j\\
&=\frac{1}{\left(2\pi\right)^{\frac{3}{2}}}\int_{\mathbb{R}^3}\hat{f}(\vb*{k})\left|\sum_{j=1}^{N}e^{i \vb*{k}\cdot\vb*{r}_j}\psi_j\right|^2\geq0,
\end{align*}
where the last line follows from
\begin{equation}
\hat{f}(\vb*{k})=\sqrt{\frac{\pi}{2}}\frac{\delta\left(k_a-\|\vb*{k}\|\right)}{k_a^2}\geq0.
\end{equation}
This calculation shows that, in analysing the eigenvalue distribution of the ERM $S$, the  idea of replacing $S$ with a simpler random matrix whose entries are $\sinc\xi_{ij}$ with $\xi_{ij}=\xi_{ji}$ \emph{independent} with the same distribution of $k_a\|\vb*{r}_i-\vb*{r}_j\|$ cannot work. The resulting random matrix will \emph{not} be non-negative definite, in general. In fact, one expects the mean eigenvalue density to be semicircular in this case.

\section{Joint distribution of the interatomic distances}
\label{app:joint}

The distributions of interatomic distance~\eqref{eq:density_R} and~\eqref{eq:density_RR'} are specialisations of  the following cute formula.
\begin{lem}Let $\vb*{X}_0,\vb*{X}_1,\ldots, \vb*{X}_k$ be independent standard Gaussian vectors in $\mathbb{R}^3$.%
 The joint density of the Euclidean lengths ${X}_0:=\|\vb*{X}_0\|,R_{01}:=\|\vb*{X}_1-\vb*{X}_0\|,\ldots, R_{0k}=\|\vb*{X}_k-\vb*{X}_0\|$  is
 \begin{multline}
 \label{eq:jointX0R}
 p({x}_0,{r}_1,\ldots,{r}_k)
 =\left(\frac{2}{\pi }\right)^{\frac{k+1}{2}}x_0^{2-k}e^{-\frac{1}{2}(k+1){x}_0^2} \\
 \times\prod_{i=1}^kr_ie^{-\frac{1}{2}{r}_i^2}
\sinh(x_0{r}_i).
\end{multline}
\end{lem}
\begin{proof}
 The joint density of the $k+1$ vectors $\vb*{X}_0,\vb*{X}_1,\ldots, \vb*{X}_k$, is
 $$
 p(\vb*{x}_0,\vb*{x}_1,\ldots, \vb*{x}_k)=\frac{1}{\left(2\pi\right)^{\frac{3}{2}(k+1)}}e^{-\frac{1}{2}\sum_{i=0}^k\|\vb*{x}_i\|^2 }.
 $$
Hence, the vectors $\vb*{X}_0$, $\vb*{R}_{01}=\vb*{X}_1-\vb*{X}_0$, $\vb*{R}_{02}=\vb*{X}_2-\vb*{X}_0$, etc. 
are jointly Gaussian with density 
\begin{align*}
 p(\vb*{x}_0,\vb*{r}_1,\ldots,\vb*{r}_k)&
 &=\frac{e^{-\frac{1}{2}(k+1)\|\vb*{x}_0\|^2-\frac{1}{2}\sum_{i=1}^k\left(\|\vb*{r}_i\|^2+2\vb*{x}_0\cdot\vb*{r}_i\right)}
}{\left(2\pi\right)^{\frac{3}{2}(k+1)}}.
 \end{align*}
 The joint density of the Euclidean lengths $X_0,R_{01},\ldots, R_{0k}$  is obtained by integrating over the angles,
 \begin{multline*}
 p({x}_0,{r}_1,\ldots,{r}_k)
 = \frac{x_0^2r_1^2\cdots r_k^2}{\left(2\pi\right)^{\frac{3}{2}(k+1)}}e^{-\frac{1}{2}(k+1){x}_0^2-\frac{1}{2}\sum_{i=1}^k{r}_i^2}\\
\small\times4\pi
 \prod_{j=1}^k2\pi \int_{0}^{\pi}d\theta_j\sin\theta_j e^{x_0{r}_j\cos\theta_j}.
\end{multline*}
After performing the elementary angular integrations we get the claimed formula. 
\end{proof}

\section{Moments of the  entries of $S$}
\label{app:mom}

The moments of the off-diagonal entries of $S$ are 
\begin{equation}
\langle S_{ij}^m\rangle=\int_{0}^{\infty}\left(\frac{\sin{\sqrt{M}r}}{\sqrt{M}r}\right)^mp_R(r)dr
\label{eq:integral_sinc}
\end{equation}
where the density~\eqref{eq:density_R} of $R=\|\vb*{x}_i-\vb*{x}_j\|$, $i\neq j$ is
$$
p_{R}(r)=\frac{1}{\sqrt{4\pi }}r^2 e^{-\frac{r^2}{4}},\quad r\geq0.
$$
It is possible to compute the first three moments,
\begin{align}
\langle S_{ij}\rangle&=e^{-M} \\
\langle S_{ij}^2\rangle&=\frac{1-{e^{-4M}}}{4M}\\ 
\langle S_{ij}^3\rangle&=\frac{\sqrt{\pi}}{2}\frac{2-\operatorname{Erfc}(\sqrt{M})+\operatorname{Erfc}(3\sqrt{M})}{8M^{\frac{3}{2}}},
\end{align}
and get the large-$M$ asymptotics
\begin{equation}
\langle S_{ij}\rangle\sim 0, \quad \langle S_{ij}^2\rangle\sim \frac{1}{4M},\quad  \langle S_{ij}^3\rangle\sim \frac{\sqrt{\pi}}{8M^{\frac{3}{2}}}.
\end{equation}
For $m\geq4$, it is difficult to perform exactly the integral~\eqref{eq:integral_sinc}. However we can extract its precise large-$M$ asymptotics. Indeed, for $m\geq4$,
\begin{align}
\langle S_{ij}^m\rangle&
=\int \frac{1}{\sqrt{M}}p_{R}\left(\frac{r}{\sqrt{M}}\right)\left(\frac{\sin r}{r}\right)^m dr \nonumber\\
&\sim\frac{1}{M^\frac{3}{2}}\left[\frac{1}{\sqrt{4\pi }}\int_0^{+\infty} r^2\left(\frac{\sin r}{r}\right)^mdr\right],
\end{align}
by dominated convergence, as $M\to\infty$. Note that \emph{all} moments for $m\geq3$ scale as $\langle S_{ij}^m\rangle\sim a_m M^{-3/2}$ for large $M$. 

The constants 
\begin{equation}
a_m=\frac{1}{\sqrt{4\pi }}\int_0^{+\infty} r^2\left(\frac{\sin r}{r}\right)^mdr,
\end{equation}
can be computed using a method of Michell and Hardy~\cite{Abel23}. The explicit result is 
\begin{equation}
a_m=\dfrac{\sqrt{\pi } \sum _{k=0}^{\left\lfloor \frac{m}{2}\right\rfloor } (-1)^{k+1} \binom{m}{k}(m-2 k)^{m-3} }{2^{m+1} (m-3)!}.
\end{equation}


\begin{thebibliography}{90}

\bibitem{Dicke}  R. H. Dicke, \textit{Coherence in Spontaneous Radiation Processes}, Phys. Rev. \textbf{93}, 99 (1954).

\bibitem{exp_sub} W. Guerin, M. O. Ara\'ujo, and R. Kaiser, \textit{Subradiance in a Large Cloud of Cold Atoms}, Phys. Rev. Lett. \textbf{116}, 083601 (2016).

\bibitem{exp_sup} M. O. Ara\'ujo, I. Kre\v{s}i\'c, R. Kaiser, and W. Guerin, \textit{Superradiance in a Large and Dilute Cloud of Cold Atoms in the Linear-Optics Regime}, Phys. Rev. Lett. \textbf{117}, 073002 (2016).

\bibitem{b0} T. Bienaim\'e, N. Piovella, and R. Kaiser, \textit{Controlled Dicke Subradiance from a Large Cloud of Two-Level Systems}, Phys. Rev. Lett. \textbf{108}, 123602 (2012).

\bibitem{essay} M. Gross and S. Haroche, \textit{Superradiance: An essay on the theory of collective spontaneous emission}, Phys. Rep. \textbf{93}, 301 (1982).

\bibitem{intro} M. O. Scully and A. A. Svidzinsky, \textit{The Super of Superradiance}, Science \textbf{325}, 1510 (2009).

\bibitem{sub_kaiser2021} A. Cipris, N. A. Moreira, T. S. do Espirito Santo, P. Weiss, C. J. Villas-Boas, R. Kaiser, W. Guerin, and R. Bachelard, \textit{Subradiance with saturated atoms: population enhancement of the long-lived states}, Phys. Rev. Lett. \textbf{126}, 103604 (2021).

\bibitem{sub2_kaiser2021} Y. A. Fofanov, I. M. Sokolov, R. Kaiser, and W. Guerin, \textit{Subradiance in dilute atomic ensembles: Role of pairs and multiple scattering}, Phys. Rev. A \textbf{104}, 023705 (2021).

\bibitem{direct2} N. E. Rehler and J. H. Eberly, \textit{Superradiance}, Phys. Rev. A \textbf{3}, 1735 (1971).

\bibitem{lamb} R. R\"ohlsberger, K. Schlage, B. Sahoo, S. Couet, and R. R\"uffer, \textit{Collective Lamb Shift in Single-Photon Superradiance}, Science \textbf{328}, 1248 (2010).

\bibitem{lamb2} J. Keaveney, A. Sargsyan, U. Krohn, I. G. Hughes, D. Sarkisyan, and C. S. Adams, \textit{The cooperative Lamb shift in an atomic nanolayer}, Phys. Rev. Lett. \textbf{108}, 173601 (2012).

\bibitem{review_coop} W. Guerin, M. T. Rouabah, and R. Kaiser, \textit{Light interacting with atomic ensembles: collective, cooperative and mesoscopic effects}, J. Mod. Opt., \textbf{64}, 895 (2017).

\bibitem{interfaces} K. Hammerer, A. S. S{\o}rensen, and E. S. Polzik, \textit{Quantum interface between light and atomic ensembles}, Rev. Mod. Phys. \textbf{82}, 1041 (2010).

\bibitem{laser} J. G. Bohnet, Z. Chen, J. M. Weiner, D. Meiser, M. J. Holland, and J. K. Thompson, \textit{A steady-state superradiant laser with less than one intracavity photon}, Nature \textbf{484}, 78 (2012).

\bibitem{quantum_internet} H. J. Kimble, \textit{The quantum internet}, Nature \textbf{453}, 1023 (2008).

\bibitem{appl_super} L.-M. Duan, M. D. Lukin, J. I. Cirac, and P. Zoller, \textit{Long-distance quantum communication with atomic ensembles and linear optics}, Nature \textbf{414}, 413 (2001).

\bibitem{appl_sub} M. O. Scully, \textit{Single Photon Subradiance: Quantum Control of Spontaneous Emission and Ultrafast Readout}, Phys. Rev. Lett. \textbf{115}, 243602 (2015).

\bibitem{K} T. Bienaim\'e, M. Petruzzo, D. Bigerni, N. Piovella, and R. Kaiser, \textit{Atom and photon measurement in cooperative scattering by cold atoms}, J. Mod. Opt. \textbf{58}, 1942 (2011).

\bibitem{skip_goetschy2} S. E. Skipetrov and A. Goetschy, \textit{Eigenvalue distributions of large Euclidean random matrices for waves in random media}, J. Phys. A: Math. Theor. \textbf{44}, 065102 (2011).

\bibitem{skip_goetschy} A. Goetschy and S. E. Skipetrov, \textit{Euclidean random matrices and their applications in physics}, arXiv preprint, arXiv:1303.2880v1 (2013).

\bibitem{Cohen1} C. Cohen-Tannoudji, J. Dupont-Roc, and G. Grynberg, \textit{Photons and Atoms: Introduction to Quantum Electrodynamics} (Wiley-Interscience, New York, 1989). 

\bibitem{Cohen} C. Cohen-Tannoudji, J. Dupont-Roc, and G. Grynberg, \textit{Atom-Photon Interactions: Basic Processes and Applications} (Wiley, New York, 1992).

\bibitem{vectorial} L. Bellando, A. Gero, E. Akkermans, and R. Kaiser, \textit{Roles of cooperative effects and disorder in photon localization: the case of a vector radiation field}, Eur. Phys. J. B \textbf{94}, 49 (2021).

\bibitem{Weiss2019}
P. Weiss, A. Cipris, M. O. Ara\'ujo, R. Kaiser, and W. Guerin, \textit{Robustness of Dicke subradiance against thermal decoherence}, Phys. Rev. A \textbf{100}, 033833 (2019).

\bibitem{Courteille} Ph. W. Courteille, S. Bux, E. Lucioni, K. Lauber, T. Bienaim\'e, R. Kaiser, and N. Piovella, \textit{Modification of radiation pressure due to cooperative scattering of light}, Eur. Phys. J. D \textbf{58}, 69 (2010).

\bibitem{spettroS} E. Akkermans, A. Gero, and R. Kaiser, \textit{Photon Localization and Dicke Superradiance in Atomic Gases}, Phys. Rev. Lett. \textbf{101}, 103602 (2008).

\bibitem{parisi} M. M\'ezard, G. Parisi, and A. Zee, \textit{Spectra of Euclidean Random Matrices}, Nucl. Phys. B \textbf{559}, 689 (1999).

\bibitem{Bogomolny03}  E. Bogomolny, O. Bohigas, C. Schmit, 
\textit{Spectral properties of distance matrices}, 
J. Phys. A: Math. Gen. {\bf 36 (12)}, 3595-3616 (2003).

\bibitem{Parisi06} G. Parisi,
\textit{Euclidean random matrices: solved and open problems}, 
In: \'E Br\'ezin, V. Kazakov, D. Serban, P. Wiegmann, A. Zabrodin (eds),
\textit{Applications of Random Matrices in Physics}, 
NATO Science Series II: Mathematics, Physics and Chemistry {\bf  221}, Springer, Dordrecht, 2006.

\bibitem{Jiang15} T. Jiang,
\textit{Distributions of eigenvalues of large Euclidean matrices generated from $l_p$ balls and spheres},
Linear Algebra and its Applications {\bf 473}, 14-36  (2015) .

\bibitem{Koltchinskii}
V. Koltchinskii and E.  Gin\'e, 
\textit{Random matrix approximation of spectra of integral operators}, Bernoulli \textbf{6}(1), 113-167  (2000).

\bibitem{Wigner2} E. P. Wigner, \textit{On the Distribution of the Roots of Certain Symmetric Matrices}, Ann. Math. \textbf{67}, 325 (1958).

\bibitem{congettura} M. V. Berry, and M. Tabor, \textit{Level clustering in the regular spectrum}, Proc. R. Soc. Lond. A \textbf{356}, 375 (1977).

\bibitem{proof_congettura} J. Marklof, \textit{The Berry-Tabor conjecture}, in Proceedings of the 3rd European Congress of Mathematics (Barcelona 2000), Progr. Math. \textbf{202}, 421 (2001).

\bibitem{BGS} O. Bohigas, M. J. Giannoni, and C. Schmit, \textit{Characterization of Chaotic Quantum Spectra and Universality of Level Fluctuation Laws}, Phys. Rev. Lett. \textbf{52}, 1 (1984).

\bibitem{mehta} M. L. Mehta, \textit{Random Matrices} (Revised and Enlarged Second Edition, Academic Press, 1991).

\bibitem{brody_1973} T. A. Brody, \textit{A statistical measure for the repulsion of energy levels}, Lett. Nuovo Cimento \textbf{7}, 482 (1973).

\bibitem{transition} T. Prosen, and M. Robnik, \textit{Energy level statistics in the transition region between integrability and chaos}, J. Phys.s A: Math. Gen. \textbf{26}, 2371 (1993).

\bibitem{transition2} T. Prosen, and M. Robnik, \textit{Semiclassical energy level statistics in the transition region between integrability and chaos: transition from Brody-like to Berry-Robnik behavior}, J. Phys.s A: Math. Gen. \textbf{27}, 8059 (1994).

\bibitem{brody} D. Engel, J. Main, and G. Wunner, \textit{Higher-order energy level spacing distributions in the transition region between regularity and chaos}, J. Phys. A: Math. Gen. \textbf{31}, 6965 (1998).

\bibitem{intermediate} E. B. Bogomolny, U. Gerland, and C. Schmit, \textit{Models of intermediate spectral statistics}, Phys. Rev. E \textbf{59}, R1315 (1999).

\bibitem{intermediate2} W.-J. Rao, \textit{Higher-order level spacings in random matrix theory based on Wigner's conjecture}, Phys. Rev. B \textbf{102}, 054202 (2020).

\bibitem{toep} E. Bogomolny, \textit{Spectral statistics of random Toeplitz matrices}, Phys. Rev. E \textbf{102}, 040101 (2020).

\bibitem{loc-deloc} F. Alet, and N. Laflorencie, \textit{Many-body localization: An introduction and selected topics}, C. R. Phys. \textbf{19}, 498 (2018).

\bibitem{loc-deloc2} F. Borgonovi, F. M. Izrailev, L. F. Santos, and V. G. Zelevinsky, \textit{Quantum chaos and thermalization in isolated systems of interacting particles}, Phys. Rep. \textbf{626}, 1 (2016).

\bibitem{porter_thomas} C. E. Porter and R. G. Thomas, \textit{Fluctuations of Nuclear Reaction Widths}, Phys. Rev. \textbf{104}, 483 (1956).

\bibitem{multifractal} F. Evers and A. D. Mirlin, \textit{Anderson transitions}, Rev. Mod. Phys. \textbf{80}, 1355 (2008).

\bibitem{guerin2017population} W. Guerin and R. Kaiser, \textit{Population of collective modes in light scattering by many atoms}, Phys. Rev. A \textbf{95}, 053865 (2017).

\bibitem{KPZ} M. Kardar, G. Parisi, and Y.-C. Zhang, \textit{Dynamic Scaling of Growing Interfaces}, Phys. Rev. Lett. \textbf{56}, 889 (1986).

\bibitem{KPZ2} I. Corwin, \textit{The Kardar-Parisi-Zhang equation and universality class}, Random Matrices: Theory and Applications Vol. 01, No. 01, 1130001 (2012).

\bibitem{KPZ3} T. Sasamoto, \textit{The 1D Kardar–Parisi–Zhang equation: Height distribution and universality}, Progress of Theoretical and Experimental Physics, Vol. 2016, Issue 2, 022A01 (2016).

\bibitem{lonigro2021stationary} D. Lonigro, P. Facchi, S. Pascazio, F. V. Pepe, D. Pomarico, \textit{Stationary excitation waves and multimerization in arrays of quantum emitters}, New J. Phys. \textbf{23}, 103033  (2021).

\bibitem{Bordenave13} C. Bordenave,
\textit{On Euclidean random matrices in high dimension},
Electron. Commun. Probab. {\bf 18}, 1-8 (2013).

\bibitem{Abel23} U. Abel, and V.  Kushnirevych, \textit{Sinc integrals revisited},
Math Semesterber, 10.1007/s00591-023-00342-5 (2023).
   \end{thebibliography}
\end{document}